\documentclass[11pt , reprint,
superscriptaddress,
nofootinbib,onecolumn,notitlepage,
amsmath,amssymb,
aps,longbibliography
]{revtex4-1}
\usepackage[english]{babel}
\usepackage[utf8x]{inputenc}
\usepackage[T1]{fontenc}
\usepackage{circuitikz}
\AtBeginDocument{%
    \newwrite\bibnotes
    \def\bibnotesext{Notes.bib}
    \immediate\openout\bibnotes=\jobname\bibnotesext
    \immediate\write\bibnotes{@CONTROL{REVTEX41Control}}
    \immediate\write\bibnotes{@CONTROL{%
    apsrev41Control,author="08",editor="1",pages="1",title="0",year="0"}}
     \if@filesw
     \immediate\write\@auxout{\string\citation{apsrev41Control}}%
    \fi
  }%

\usepackage[nodayofweek]{datetime}
\usepackage{amsmath}
\usepackage{url}
\usepackage{graphicx}
\usepackage{amsthm}
\usepackage{dcolumn}
\usepackage{bm}
\usepackage{amsmath}
\usepackage{amssymb}
\usepackage{amsfonts}
\usepackage{graphicx}
\usepackage{lipsum}

\usepackage{hyperref}
\usepackage{xcolor}
\definecolor{mylinkcolor}{rgb}{0,0,0.8} 
\hypersetup{unicode=true,%
  bookmarksnumbered=false,bookmarksopen=false,bookmarksopenlevel=1, %
  breaklinks=true,pdfborder={0 0 0},colorlinks=true}%
\hypersetup{%
  anchorcolor=mylinkcolor,citecolor=mylinkcolor, %
  filecolor=mylinkcolor,linkcolor=mylinkcolor, %
  menucolor=mylinkcolor,runcolor=mylinkcolor, %
  urlcolor=mylinkcolor}

\newtheorem{lemma}{Lemma}
\newtheorem{theorem}{Theorem}

\theoremstyle{definition}
\newtheorem{definition}{Definition}

\newtheorem{protocol}{Protocol}

\newcommand{\tr}{\mathrm{tr}}
\newcommand{\ket}[1]{| #1 \rangle}
\newcommand{\bra}[1]{\langle #1 |}
\newcommand{\braket}[2]{\langle #1|#2\rangle}
\newcommand{\ketbra}[2]{\ket{#1}\!\bra{#2}}
\newcommand{\proj}[1]{\ket{#1}\!\bra{#1}}
\newcommand{\id}{\openone}

\newcommand{\cI}{\mathcal{I}}

\newcommand{\ot}{\otimes}
\newcommand{\bin}{\mathrm{bin}}
\newcommand{\scorefunction}{S}
\newcommand{\score}{\omega}
\newcommand{\Ext}{\mathrm{Ext}}

\newcommand{\rate}{\mathrm{rate}}

\newcommand{\overlap}{\Theta}
\newcommand{\overlapfunction}{\mathcal{O}}
\newcommand{\entropy}{\mathcal{G}}
\newcommand{\entropytwo}{\mathcal{P}_1}

\newcommand{\comment}[1]{}
\def\M{\ensuremath\mathcal}
\def\B{\ensuremath\mathbf}
\usepackage{enumitem}
\usepackage{mdframed}
\begin{document}
\title{Higher rates for semi-device-independent randomness expansion by recycling input randomness}
\author{Rutvij Bhavsar}
\email{rutvij.bhavsar@kcl.ac.uk}
\affiliation{Department of Mathematics, King's College London, Strand, London, WC2R 2LS, United Kingdom}
\affiliation{Department of Mathematics, University of York, Heslington, York, YO10 5DD, United Kingdom}
\affiliation{School of Electrical Engineering, Korea Advanced Institute of Science and Technology (KAIST),
291 Daehak-ro, Yuseong-gu, Daejeon 34141, Republic of Korea}
\author{Hamid Tebyanian}
\email{h.tebyanian@qmul.ac.uk}
\affiliation{School of Physical and Chemical Sciences, Queen Mary University of London, London, E1 4NS, United Kingdom}
\affiliation{Department of Mathematics, University of York, Heslington, York, YO10 5DD, United Kingdom}
\author{Roger Colbeck}
\email{roger.colbeck@kcl.ac.uk}
\affiliation{Department of Mathematics, King's College London, Strand, London, WC2R 2LS, United Kingdom}
\affiliation{Department of Mathematics, University of York, Heslington, York, YO10 5DD, United Kingdom}
\date{\today}
\begin{abstract}
Although quantum random number generators rely on the inherent indeterminism of quantum mechanics, ensuring that the numbers produced are secure remains a significant challenge. We introduce two semi-device-independent randomness expansion protocols in a prepare-and-measure setting, where the source and measurement devices are treated as uncharacterised and we assume trust only in testing device, which could be implemented using a photodiode. One protocol achieves expansion by recycling the input randomness, while the other uses a biased input distribution to achieve expansion in settings where recycling is not possible. The protocols are proven secure against quantum side information. Our results show that high randomness rates are achievable under experimentally realistic conditions, with expansion obtained in as few as $10^5$ to $10^6$ rounds with the recycling protocol.
\end{abstract}

\maketitle

\section{Introduction} 

Random numbers serve diverse purposes, spanning cryptography, gambling, and scientific experiments. For cryptographic applications, there are two key requirements: the numbers in the output should be uniformly distributed and they should be unguessable by any third party. Protocols for randomness expansion aim to take some initial seed randomness and generate a longer string of output random numbers satisfying both of these requirements.

The standard quantum way to generate random numbers relies on having a detailed physical model of the devices used in the protocol. Security is proved based on the model and if the real world device deviates from it, the security proof may no longer apply. Furthermore, the properties of real-world devices can change over time due to factors such as temperature changes, component degradation or even tampering. To address this, device-independent (DI) randomness expansion protocols~\cite{ColbeckThesis,CK2,PAMBMMOHLMM} were introduced. These protocols achieve randomness expansion without relying on any assumptions about the inner workings of the devices used in the protocol, removing the need for a detailed device model and hence protecting against the aforementioned problems\footnote{For instance, components may still degrade, but the protocol will automatically abort if the devices are not working sufficiently well to generate secure output.}. However, implementing DI protocols experimentally poses significant challenges, primarily due to the need for robust entanglement distribution and detection. Consequently, while DI protocols offer very high security, current technology limits their applicability to only a few select applications~\cite{LLR&,Shalm2021}.

By using physically well-motivated assumptions about the experimental setup and the devices, semi-device-independent protocols aim to achieve high levels of security while being significantly easier to implement and having higher randomness rates compared to DI protocols. Since distributing entanglement presents a primary challenge when implementing DI protocols, most semi-DI protocols use a prepare-and-measure setup~\cite{s22,semiDI-Li,wang2023provably}, which is generally more feasible to implement experimentally.

In this work, we propose semi-device-independent (semi-DI) randomness expansion protocols based on the prepare-and-measure framework introduced in~\cite{VSoriginal}, which has since been studied for randomness certification tasks in semi-DI settings~\cite{VS, RVMBWPB, teb_cv, Brask2017, Rusca2020, avesani2020, Tebyanian_2021}\footnote{It may be possible to modify some of the existing protocols to not only certify randomness, but also expand it. This would require modifications to the protocols, especially in terms of the input distribution.}. Unlike certification protocols, which aim to verify the unpredictability of a given string, our focus is on randomness expansion -- that is, generating an output string whose length exceeds the amount of initial randomness consumed. This task becomes particularly relevant in scenarios where the input randomness is a private bit string and cannot be treated as a free resource.

Many existing semi-device-independent randomness expansion protocols, such as source-device-independent~\cite{cao2016source} and measurement-device-independent schemes~\cite{wang2023provably}, rely on partial characterisation of either the source or the measurement device. In contrast, our assumption model is weaker: we do not characterise \emph{either} the source or the measurement device, and instead require trust in a testing device, which could be implemented using a photodiode. This testing device acts as a binary on/off detector, whose behaviour is significantly easier to model and verify in practice compared with a power meter that would be required to justify energy-based assumptions as in~\cite{VS,VSoriginal}. All the quantities needed in our analysis are inferred directly from the \textit{in-protocol} detection statistics, rather than imposed as external assumptions that would require running a separate composable protocol to certify source properties~\cite{Bhavsar2025Composable}. Running a separate composable state-certification protocol using standard techniques is itself resource intensive~\cite{wiesner2024quantum}.

We establish security using the Entropy Accumulation Theorem (EAT)~\cite{DFR,DF,MFR}, which provides security against adversaries holding quantum side information and avoids any i.i.d.\ assumption in the protocol, in contrast to many existing semi-DI randomness certification and expansion protocols.A precise statement of the physical and adversarial assumptions underlying our model is provided in Section~\ref{sec: security assumptions}.

Our main protocol achieves randomness expansion by recycling input randomness: part or all of the seed used in the extractor is fed back into the next round of the protocol. This substantially reduces the seed requirement while preserving composable security, and, as we show, enables high expansion rates even at relatively small numbers of rounds. For completeness, we also describe a second variant based on heavily biased input choices, in the spirit of spot‑checking protocols in the DI setting (see, e.g.,~\cite{LLR&}). This variant may be of interest in scenarios where the input randomness cannot be recycled, for instance when it originates from a public randomness beacon.

Given the minimal physical assumptions in our framework, one might expect a significant loss in performance. However, our results show that this is not the case. The recycling protocol achieves strong asymptotic expansion rates under realistic experimental parameters, and its performance is particularly notable in the finite‑size regime. Positive expansion can be obtained with as few as $10^5$ rounds in favourable experimental conditions, and with approximately $10^6$ rounds even under realistic imperfections. These results demonstrate that, despite operating under strictly weaker assumptions than existing semi‑DI approaches, our protocol delivers both high certified rates and practical finite‑size efficiency.

A key reason for the improved asymptotic performance of our protocols is that we estimate the single-round von Neumann entropy directly, rather than first lower-bounding it via the min-entropy and employing semidefinite programming relaxations. This leads to tighter entropy bounds under realistic experimental parameters.

\section{General framework}\label{sec: general_semi-DI protocol}

\begin{figure}
    \centering
    \includegraphics[width=1 \textwidth]{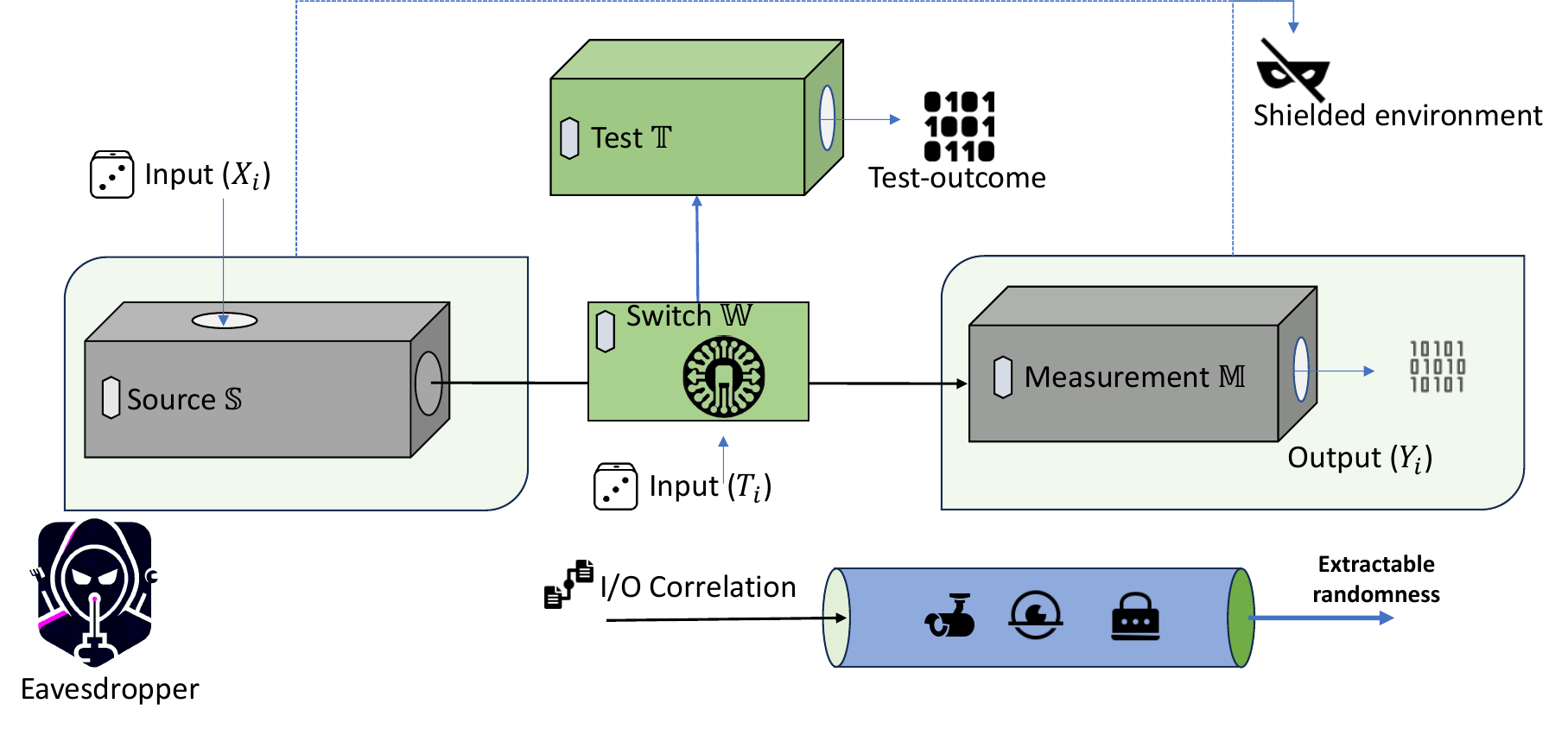}
    \caption{Schematic illustration of our protocols. The setup consists of $4$ different components: the source $\mathbb{S}$, switch $\mathbb{W}$, measurement device $\mathbb{M}$ and testing device $\mathbb{T}$. On round $i$, the source receives random input $X_i$ and generates a corresponding state which is sent to the switch. The switch takes random input $T_i$, corresponding on whether it sends the state to the trusted testing device or measurement device, where a measurement is performed giving outcome $Y_i$. The experiment is carried out in a secure lab, from which, it is assumed that information leak cannot happen. The devices in the black boxes are treated as uncharacterized, while the green components are fully trusted. Our two protocols differ only in the details of the extraction step, which also requires additional input randomness.}\label{fig: semi_DI_diag}
\end{figure}

Figure~\ref{fig: semi_DI_diag} gives a schematic illustration of the steps in our protocols. An honest implementation of the protocol is performed by a desired `honest source' \(\mathbb{S}^{\mathrm{honest}}\) and an desired `honest measurement device' \(\mathbb{M}^{\mathrm{honest}}\) (these are the devices that an experimenter would ideally like to set up). Both $\mathbb{S}^{\mathrm{honest}}$ and $\mathbb{M}^{\mathrm{honest}}$ behave in an i.i.d.\ fashion (note though that we do not assume this honest behaviour when analysing security; it is used to illustrate what we would ideally like the device to do). We describe each component of the protocol below.
\begin{itemize}
    \item \textbf{Source $\mathbb{S}$} (\textit{uncharacterized}): 
In round \(i\), an input \(X_{i} \in \{0, 1\}\) is generated according to an input probability distribution \(p_{X}\) provided to the source. Based on this input, the source prepares a state \(\rho^{x}_{i}\) (called a signal). Note that the source can prepare different states in each round (i.e., $\rho^{x}_{i}$ need not be equal to $\rho^{x}_{j}$ for $i \neq j$). \\ 
An honest source \(\mathbb{S}^{\mathrm{honest}}\) is a memoryless device that prepares a fixed pure state depending upon the input received, i.e., when $X_i=x$ the source $\mathbb{S}^{\mathrm{honest}}$ prepares $\rho^{x}_{i} = \proj{\psi^{x}}$ (independent on the round number $i$) and sends the state to the switch. 

\item \textbf{Switch $\mathbb{W}$:} The switch $\mathbb{W}$ takes an input $T_{i} \in \{ 0 , 1\}$.  If $T_{i} = 0$, then the switch sends the signal to the measurement device $\mathbb{M}$, otherwise the signal is sent to the testing device $\mathbb{T}$. The bias $p_{T}(0)$ is chosen to be approximately $1$, so that most of the time the signal is sent to the measurement device. The switch can be uncharacterized and in principle we ought to add an abort condition to the protocol for when there are too many detector clicks not aligned with the switch value. We do not consider this aspect in detail in this work because in practice good switches are straightforward to construct.

\item \textbf{Testing device $\mathbb{T}$} (\textit{characterized}): Upon receiving the signal the testing deive $\mathbb{T}$ performs a projective measurement \(\{\Pi_{0}, \id - \Pi_{0}\}\), where \(\Pi_{0} = \proj{0}\) is the projection on the ground state. If the ground state is measured, then the variable \(Y_{i} \) is set to $0$ otherwise it is set to $1$. The overlap $\Theta$ is estimated from input-output statistics of the test rounds. It is defined as the average probability of measuring the ground state:
\begin{eqnarray}
\overlap  := \frac{1}{2} \left( p(Y  = 0 | X  = 0, T  = 1) + p(Y  = 0 | X  = 1, T  = 1) \right).
\end{eqnarray}
For our protocols, since we do not assume i.i.d.\ behaviour, the quantity $\overlap$ is computed from the observed statistics collected during the protocol. Here $p(Y = y | X = x, T = t)$ are estimated by computing the average frequency of observing $Y = y$ when $X = x$ and $T = t$ in the statistics collected over many rounds.\footnote{Note that the definition of $\overlap$ above is independent of the input probability distribution $p_{X}$, i.e., the factor of $1/2$ should not be replaced with $p_{X}(x)$ when the input distribution is not uniform.}

\item \textbf{Measurement $\mathbb{M}$} (\textit{uncharacterized}): The role of the measurement device is to output a bit \(Y_{i} \in \{0, 1\}\) upon receiving the signal. If \(Y_{i} = X_{i}\), we consider the round won; otherwise, the round is lost. Similar to the test rounds, the score is estimated from input-output statistics of the measurement rounds. It is defined as the average winning probability:
\begin{eqnarray}
\score := \frac{1}{2} \left( p(Y  = 0 | X  = 0, T = 0) + p(Y  = 1 | X  = 1, T  = 0) \right).
\end{eqnarray}
As with $\overlap$ to compute $\score$ for our protocols, we estimate the conditional probabilities $p(Y = y | X = x, T = t)$ from the observed frequencies in the collected data.
An honest measurement device \(\mathbb{M}^{\mathrm{honest}}\) is a device that performs a pre-defined two-outcome measurement \(\{M_{0}, \id - M_{0}\}\) in each round.

\end{itemize}

Essentially, the protocol involves the source and the measurement device playing a state discrimination game, where the measurement device tries to guess the value $x$ of $X_i$ given $\rho_i^x$. Periodically, a ``spot-check'' is performed using the testing device to ensure that the source produces two states with significant overlap $\overlap$ with the ground state $\proj{0}$, ensuring that the states cannot be perfectly distinguished. Achieving a high enough score then ensures that the outputs $Y_i$ of the measurement device contain extractable randomness.   

Not all possible values of \((\score, \overlap)\) lead to randomness expansion. To illustrate this, we give a classical strategy that achieves overlap $\overlap$ for any $\overlap>1/2$. Suppose that when \(X = 0\) the source prepares the state \(\rho^0 = \proj{0}\), and when \(X = 1\), the source prepares \(\rho^1 = (2 \overlap - 1) \proj{0} + 2(1 - \overlap) \proj{1}\). The state \(\rho^1\) can be constructed using a classical binary random variable \(\Lambda\), where \(\Lambda = 0\) occurs with probability \(p_\Lambda(0) = 2 \overlap - 1\). Additionally, a score \(\score = 3/2 - \overlap\) can be obtained if the measurement device that performs the measurement \(\{\proj{0}, \id - \proj{0}\}\). In such cases, given the input \(X\), and access to the classical variable \(\Lambda\) the value of the bit \(Y\) can be determined with certainty. Thus, for any overlap \(\overlap>1/2\), for \(\score < \frac{3 - 2 \overlap}{2}\), the conditional entropy \(H(Y|XE)=H(XY|E) - H(X|E)\) is equal to \(0\), and hence randomness expansion cannot be achieved\footnote{Here $\Lambda$ is contained within $E$.}.

Note also that for a given overlap $\overlap$, a genuine quantum strategy exists that achieves a score of $\score = \frac{1}{2} + \sqrt{\overlap (1 - \overlap)}$. For this, the states $\rho^{0} = \proj{\psi_{0}}$, where $\ket{\psi_0} = \sqrt{\overlap} \ket{0} + \sqrt{1 - \overlap} \ket{1}$, and $\rho^{1} = \proj{\psi_{1}} = \sqrt{\overlap} \ket{0} - \sqrt{1 - \overlap} \ket{1}$ can be used. A score of $\frac{1}{2} + \sqrt{\overlap (1 - \overlap)}$ can be achieved using the measurement with $M_0 =P_+$, where $P_+$ is the projector onto the positive eigenspace of the operator $\proj{\psi_0} - \proj{\psi_1}$. The measurement $\{M_0,\id-M_0\}$ is the Helstrom measurement~\cite{hel69,hol73}, and is optimal for these states.

As the overlap, $\overlap$, approaches $1$, the fidelity between each of the prepared states and $\proj{0}$ increases, the states become more difficult to distinguish, and the maximum achievable score tends to $1/2$. For such values of $\overlap$, a small amount of experimental noise can lead to values of $\score$ that give no randomness. On the other hand, for $\overlap\approx1/2$, it is possible to obtain high scores; however, for $\overlap\approx1/2$ a wide range of scores can be achieved using a strategy that involves outputting pre-shared randomness. 
In Section~\ref{sec: results}, we provide plots (cf.\ Figure~\ref{fig:Gfunct}) showing the regime for which the protocol can be used for randomness expansion, and indicating the values of $(\score, \overlap)$ that are most useful for running the protocol.

\section{Security assumptions}\label{sec: security assumptions} 
We use the composable security definition given in Appendix~\ref{app: security definition}. The assumptions used in our work are as follows:
\begin{enumerate}
    \item\label{ass: DI-Quantum} Quantum theory is correct and complete. 
    \item\label{ass: DI-Sheilding} All protocols take place in a laboratory that is shielded from the outside world.
    \item\label{ass: carrier} The source and the measurement device communicate solely via the intended (and known) carrier of information. This carrier has an energy spectrum with the following properties:
    \begin{enumerate}
    \item\label{ass: sDI-ground} The ground state of the system prepared by the source is unique. 
    \item\label{as: sDI-gap} There is a gap between the ground state and the first excited state. 
    \end{enumerate}
    \item\label{ass: sDI-photodiode} The testing device is fully characterized.
    \item\label{ass: sDI-measurement} The measurement device does not share entanglement with any other system (the measurement device can share classical randomness with the source as well as with the adversary).
\end{enumerate}

Assumptions~\ref{ass: DI-Quantum} and~\ref{ass: DI-Sheilding} are fundamental for any randomness expansion protocol based on quantum theory, including device-independent (DI) protocols. Assumption~\ref{ass: carrier} is crucial for our protocol to prevent the source from communicating $X$ to the measurement device using an alternative carrier of information. If the measurement device could deterministically learn $X$ via the alternative carrier, it could then produce the output $Y$ using $X$ and some pre-shared classical randomness to achieve any desired score. For instance, in the case where $X$ is uniformly distributed, upon learning \(X\), the measurement device outputs \(Y = X\) with probability \(\score\) and \(Y = 1 - X\) with probability \(1-\score\). This random choice can be made using randomness that is preshared with the adversary, meaning that with access to $X$ the adversary could determine the output $Y$ with certainty. The additional carrier might not be detected by the testing device, allowing the source to prepare signals with any desired overlap. However, if light is the intended carrier, additional physical assumptions and time stamping can reasonably justify Assumption~\ref{ass: carrier}. More details are provided in Appendix~\ref{app: discussion_on_carrier_information}.

Assumption~\ref{ass: sDI-photodiode} is arguably weaker than assuming a trusted power meter (as in other works) because characterizing a power meter requires knowledge of the full energy spectrum of the system, whereas the testing device we use only needs to distinguish between the ground state and any excited state (in a photonic implementation, this can be achieved using a photodiode, for instance). We would ideally like to weaken this assumption, although it cannot be removed entirely (without another assumption replacing it).

Assumption~\ref{ass: sDI-measurement} posits that the source and measurement device do not share entanglement at any stage of the experiment. This assumption is used for theoretical convenience and ideally should be avoided. With present technology, it is difficult to store entanglement for more than a few seconds, so the source and measurement device can be isolated for a small amount of time can remove entanglement. We leave the relaxation of this assumption for future work.

\section{Protocols} \label{sec: protocols}
We now introduce two protocols for randomness expansion using this setup. The protocols have similarities with the spot-checking approach used in DI protocols. They use trusted random number generators (or some initial source of randomness) to generate inputs \(X\) and \(T\) respectively (indicated as dice in Figure~\ref{fig: semi_DI_diag}). We first consider the protocol that recycles input randomness.

\begin{center}
\begin{mdframed}[linecolor=black, roundcorner=5pt, skipabove=10pt, skipbelow=10pt, backgroundcolor=white, splittopskip=10pt, splitbottomskip=10pt]

\begin{protocol}\label{protocol: Semi-DI recycling}
\noindent\textbf{Parameters}:\\
$n$ -- number of rounds \\
$p_{0} $ -- probability that $X = 0$ (taken to be at most $1/2$)\\ 
$\gamma$ -- probability of a test (taken to be at most $1/2$)\\ 
$\overlap_{\text{exp}}$ -- expected overlap (taken to be greater than $1/2$) \\
$\delta_{\overlap}$ -- confidence width for the overlap \\
$\omega_{\text{exp}}$ -- expected score \\
$\delta_{\score}$ -- confidence width for the score \\
\begin{enumerate}
    \item\label{st: 1} Set $i=1$ for the first round, or increase $i$ by 1.
    \item Randomly choose $X_i\in\{0,1\}$, which is input to the source device $\mathbb{S}$. Here $X_i=0$ occurs with probability $p_{0}$. The device $\mathbb{S}$ sends a system to the switch $\mathbb{W}$. 
    \item Randomly choose $T_i\in\{0,1\}$, where $T_i=0$ has probability $1 - \gamma$, and input $T_i$ to $\mathbb{W}$. $\mathbb{W}$ sends the system to the measurement device $\mathbb{M}$ if $T_i = 0$ or sends it to the testing device $\mathbb{T}$ if $T_i = 1$. 
    \item \begin{enumerate}
        \item If $T_i  = 0$ (measurement round): $\mathbb{M}$ receives the system and outputs $Y_i\in\{0,1\}$.
        \item  If $T_i  = 1$ (test round): $\mathbb{W}$ receives the system and outputs $Y_i\in\{0,1\}$. 
    \end{enumerate} 
    \item Return to Step~1 unless $i=n$. 
        \item Set  $U_i = (T_i , X_i , Y_i)$ and compute the empirical scores $\overlap_{\#}$ and $\score_{\#}$ as 
        \begin{eqnarray}
        \overlap_{\#} := \frac{1}{2}\sum_{x} \frac{|\{i : U_{i} = (1 , x , 0) \}|}{n p_{X}(x)\gamma}, \quad 
        \score_{\#} := \frac{1}{2}\sum_{x}\frac{|\{i :U_{i} = (0 , x , x)\}|}{n p_{X}(x)(1 - \gamma) }  \nonumber .
        \end{eqnarray}
    \item\label{st: 8} Abort the protocol if either $\overlap_{\#} < \overlap_{\text{exp}} - \delta_{\overlap}$ or $\score_{\#} < \score_{\text{exp}} - \delta_{\score}$.
\item\label{st: 9} Process the concatenation of all the outputs with a quantum-proof strong extractor $\Ext$ to yield $\Ext(\mathbf{XYT},\mathbf{R})$, where $\mathbf{R}$ is a random seed for the extractor. Since a strong extractor is used, the final outcome can be taken to be the concatenation of $\mathbf{R}$ and $\Ext(\mathbf{XYT},\mathbf{R})$.
\end{enumerate}
\end{protocol}

\end{mdframed}
\end{center}

The final step of Protocol~\ref{protocol: Semi-DI recycling} uses both the input strings ($\mathbf{T},\mathbf{X}$) and the output string $\mathbf{Y}$ in the extractor (the use of $\mathbf{T}$ and $\mathbf{X}$ we term as recycling the input randomness). This recycling is not needed for expansion, but, if recycling is not used, the input distribution $p_X$ cannot be taken to be uniform. A uniform distribution would imply that the input randomness per round is $1$ bit, and, since the output randomness is upper bounded by $1$ bit per round, expansion would be impossible. To achieve randomness expansion without recycling, we use a heavily biased input distribution $p_X(0) \approx 0$, so that the amount of input randomness consumed per round is significantly reduced. The resulting protocol differs from Protocol~\ref{protocol: Semi-DI recycling} only in its final step, where Step~\ref{st: 9} is replaced by Step~$\ref{st: 9}'$ as follows:

\begin{center}
\begin{mdframed}[linecolor=black, roundcorner=5pt, skipabove=10pt, skipbelow=10pt, backgroundcolor=white, splittopskip=10pt, splitbottomskip=10pt] \begin{protocol}\label{protocol: private to public} 
\noindent\textbf{Parameters:} Same as Protocol~\ref{protocol: Semi-DI recycling}.
\begin{enumerate} \item[\ref{st: 1}--\ref{st: 8}] Follow the corresponding steps in Protocol~\ref{protocol: Semi-DI recycling}. \item[$\ref{st: 9}'$] Process the concatenation of all outputs using a quantum-proof strong extractor $\Ext$, yielding $\Ext(\mathbf{Y}, \mathbf{R})$, where $\mathbf{R}$ is a uniformly random seed for the extractor.
\end{enumerate}
\end{protocol} 
\end{mdframed}
\end{center}

Protocol~\ref{protocol: private to public} is most relevant when the input randomness potentially becomes known to the eavesdropper during the protocol, and thus cannot be recycled. For instance, this would occur when the input is obtained from a trusted but public randomness source, such as a randomness beacon. In such situations, the adversary may learn the input string used in the protocol\footnote{If the input is taken from a beacon, it is crucial that the devices are no longer under adversarial control at the time the input becomes publicly available.}, as it is publicly accessible. Since it is generated externally and publicly, the input randomness can be treated as a free resource (the protocol can be thought of as turning public randomness into private randomness). In this case, using a heavily biased input distribution reduces the rate of output randomness in the asymptotic regime. Therefore, it is preferable to use an unbiased input distribution (i.e., $p_X=1/2$) for Protocol~\ref{protocol: private to public}. With an unbiased input distribution, the asymptotic rates of the protocol coincide with those of Protocol~\ref{protocol: Semi-DI recycling}.

We conclude the discussion of the protocols by detailing how an honest implementation of the protocol would be. One such honest implementation was discussed in Section~\ref{sec: general_semi-DI protocol}, in which the honest source prepares qubit states $\ket{\psi_0} = \sqrt{\overlap} \ket{0} + \sqrt{1 - \overlap} \ket{1}$ when $X = 0$, and $\ket{\psi_1} = \sqrt{\overlap} \ket{0} - \sqrt{1 - \overlap} \ket{1}$ when $X = 1$. The honest measurement device performs a projective measurement $\{P_+,\id-P_+\}$, where $P_+$ is the projection onto the positive part of $\proj{\psi_0} - \proj{\psi_1}$. 

Another possible honest implementation is to have a source preparing the coherent states $\proj{\alpha}$ if $X = 0$ and $\proj{-\alpha}$ if $X = 1$. The value of $\alpha$ is determined by the desired overlap $\overlap$ by
$|\alpha| = \sqrt{\ln\left( \frac{1}{\overlap} \right)}
$. The measurement device performs the optimal measurement that gives the highest value of the $\score$ given by the Helstrom bound $\score = \frac{1}{2} + \frac{\sqrt{1 - \overlap^4}}{2}$\footnote{In practice, such optimal measurements may be very difficult to implement \cite{SW2021}.}.

\section{Randomness rates}\label{sec: rates}

In this section, we outline the general structure for the raw randomness generation component of Protocols~\ref{protocol: Semi-DI recycling} and~\ref{protocol: private to public} (a deeper discussion of the finite-round rates is in Appendix~\ref{app: asymptotic rates}). Our aim is to present an informal argument to identify the relevant entropic quantity necessary for computing the randomness rate in these protocols. Before we discuss the computation of the randomness rate, we note that the protocol consists of two types of round, in terms of which the channel describing a single round of a protocol $\M{N}$ can be expressed as 
\begin{equation*}
    \M{N} = (1 - \gamma) \M{N}^{G} \ot \proj{0}_{T} + \gamma \M{N}^{T} \ot \proj{1}_{T}. 
\end{equation*}
Here $\M{N}^{G}$ and $\M{N}^{T}$ are the channels corresponding to the measurement and the test rounds, $\gamma$ is the testing probability, and $T$ is the classical register that records the input to the switch (i.e., whether a test is being performed or not). The detailed mathematical description of the channels $\M{N}^{G}$ and $\M{N}^{T}$ is given in Appendix~\ref{app: EAT channels}.

Computing the rates for our protocols is challenging, primarily due to the lack of structure in the problem that arises from the desire not to make many assumptions about how the devices operate. We need to account for arbitrary preparations, arbitrary measurements, and potentially adaptive strategies between rounds. The Entropy Accumulation Theorem (EAT)~\cite{DFR,ARV} reduces the complexity of the problem to that of computing the single-round von Neumann entropy of the outputs conditioned on the side information held by the adversary, as a function of the observed experimental statistics. Importantly, the EAT is tight in the asymptotic limit, and justifies using the single-round von Neumann entropy as the asymptotic rate of the protocol. The EAT is equipped to consider quantum side information and does not assume an i.i.d.\ behaviour when proving.

The EAT can be applied to compute the finite-round rates of the protocols by first determining the asymptotic rate and subtracting a penalty term that scales as $1/\sqrt{n}$, where $n$ is the number of rounds. The finite-round rates can be computed in terms of the single-round von Neumann entropy by constructing a \emph{min-tradeoff function}, which, roughly speaking, is the gradient of the rate function (i.e., the single-round von Neumann entropy bounds over all possible achievable values of $\score$ and $\overlap$, not necessarily those observed in the experiment). The explicit computation of the min-tradeoff function can be performed using existing techniques such as those presented in~\cite{BRC}. Thus, our results can be directly extended to enable finite-round analysis.

We now focus our discussion on the rates for individual protocols. For Protocol~\ref{protocol: Semi-DI recycling}, since both the input and output strings are used in the extraction step, the difference between the output randomness and the input randomness (in the asymptotic limit) is given by
\begin{eqnarray}
  \text{rand}_{\text{out}} - \text{rand}_{\text{in}}&=&  H(TXY|E) - H(TX) \nonumber\\
                                                    &=& p_{T}(0)H(Y|X,T= 0,E) +   p_{T}(1)H(Y|X, T=1 , E) \nonumber\\ 
                                                    &\geq&(1-\gamma)H(Y|X,T= 0,E) \label{eqn: rand_expansion_recycled}
\end{eqnarray}
where we have used the chain rule for the conditional von Neumann entropy and independence of $X$, $T$ and $E$ (by the assumption that $X$ and $T$ are chosen using perfect random number generators). We compute the bounds on the von Neumann entropy $H(Y|X,T=0,E)$ in Appendix~\ref{app: rates}. Note that the entropy $H(Y|X,T=0,E)$ in~\eqref{eqn: rand_expansion_recycled} is the randomness generated in the measurement rounds; a way to increase the rates would to use randomness generated in test rounds as well. Further, note that~\eqref{eqn: rand_expansion_recycled} does not contain a term corresponding to the random seed needed for the extractor. This is because if a strong quantum-proof extractor is applied, the seed required remains random after its use and is therefore not consumed. (Technically, the seed degrades by a very small amount -- a detailed discussion on this is given in~\cite{LLR&}.)

In Protocol~\ref{protocol: private to public}, we do not recycle the input randomness, so the asymptotic randomness expansion rate is given by 
\begin{eqnarray}
     \text{rand}_{\text{out}} - \text{rand}_{\text{in}} \geq (1 - \gamma)H(Y|T=0,E) - H(X) - H(T) .
\end{eqnarray} 
If Protocol~\ref{protocol: private to public} is used as a randomness certification protocol, then the amount of private randomness certified in the asymptotic limit becomes  
\begin{eqnarray}\label{eqn: rand_expansion}
\text{rand}_{\text{out}}  \geq (1 - \gamma)H(Y|X,T=0,E).
\end{eqnarray} 
Note that for both protocols, taking $\gamma \rightarrow 0$ is valid in the asymptotic limit, since a finite number of test rounds suffices to estimate $\overlap$ with arbitrarily high statistical confidence. As a result, $\gamma$ is not required in the asymptotic-rate plots presented later in this work. However, $\gamma$ plays a role in the finite-round analysis, as the confidence on $\overlap$ becomes statistical and must be accounted for when deriving finite-round bounds.

\section{Results}\label{sec: results}

Before presenting the rates of the protocol, we first focus on the rates obtained by restricting to qubit strategies (i.e., when the source prepares a qubit state and the measurement device performs a projective measurement). We denote the asymptotic rates achieved using qubit strategies by the function \( G_{p_X}(\score, \overlap) \). Figure \ref{fig:Gfunct} illustrates the behaviour of \( G_{p_X}(\score, \overlap) \) for two extremal input distributions: \( p_X(0) = 1/2 \), relevant for Protocol~\ref{protocol: Semi-DI recycling}, and \( p_X(0) = 1 - 10^{-2} \), relevant for Protocol~\ref{protocol: private to public}. 

As discussed in Appendix~\ref{app: reducing the strategy space}, the function \( G_{p_X}(\score, \overlap) \) also directly relates to achievable asymptotic rates. In particular, it provides the asymptotic rates achievable (without any restriction on the dimension of the system prepared by the source) with any honest measurement device\footnote{This is differs from Assumption \ref{ass: sDI-measurement}, which allows for the measurement device and the source to share classical correlations.} that performs a projective measurement. The function $G_{p_X}(\score, \overlap)$ contains regions in the parameter space of \((\score, \overlap)\) that are incompatible with any quantum strategy (shown as gray-shaded areas in Figure~\ref{fig:Gfunct}), as well as regions where no randomness can be certified.

\begin{figure}[h!]
    \centering
    \includegraphics[width=0.95\textwidth]{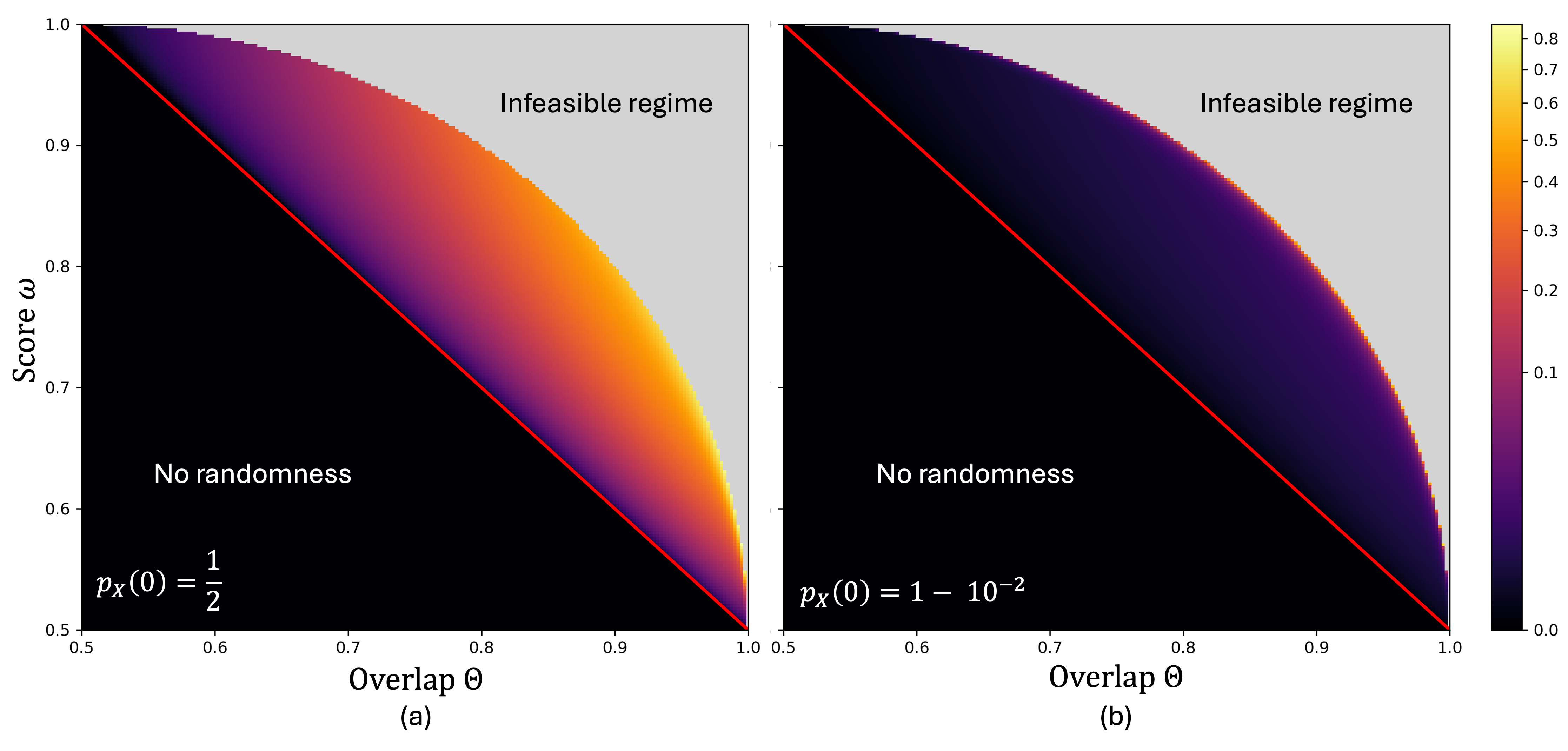}
    \caption{Lower bounds on the function \( G_{p_X}(\score, \overlap) \) for different values of \( p_X \). Panel (a) corresponds to the uniform input distribution relevant for Protocol~\ref{protocol: Semi-DI recycling}, while panel (b) pertains to the heavily biased input distribution associated with Protocol~\ref{protocol: private to public}. The shaded gray region represents parameter regimes where no quantum strategies achieving the corresponding \(\score\) and \(\overlap\) values exist. The red lines indicate $\score=(3 - 2\overlap)/2$; as discussed in Section~\ref{sec: general_semi-DI protocol}, below this line classical ``mixing'' strategies exist, so there cannot be any randomness.}
    \label{fig:Gfunct}
\end{figure}

The asymptotic randomness rate for the protocols (under the assumptions taken in Section~\ref{sec: security assumptions}) can be computed by taking the convex envelope (or convex lower bound) of \( G_{p_X}(\score, \overlap) \) (see Appendix~\ref{app: rates}; see also Appendix~\ref{app: LF transform} for the definition of the convex envelope).

\begin{figure}[h!]
    \centering
    \includegraphics[width=\textwidth]{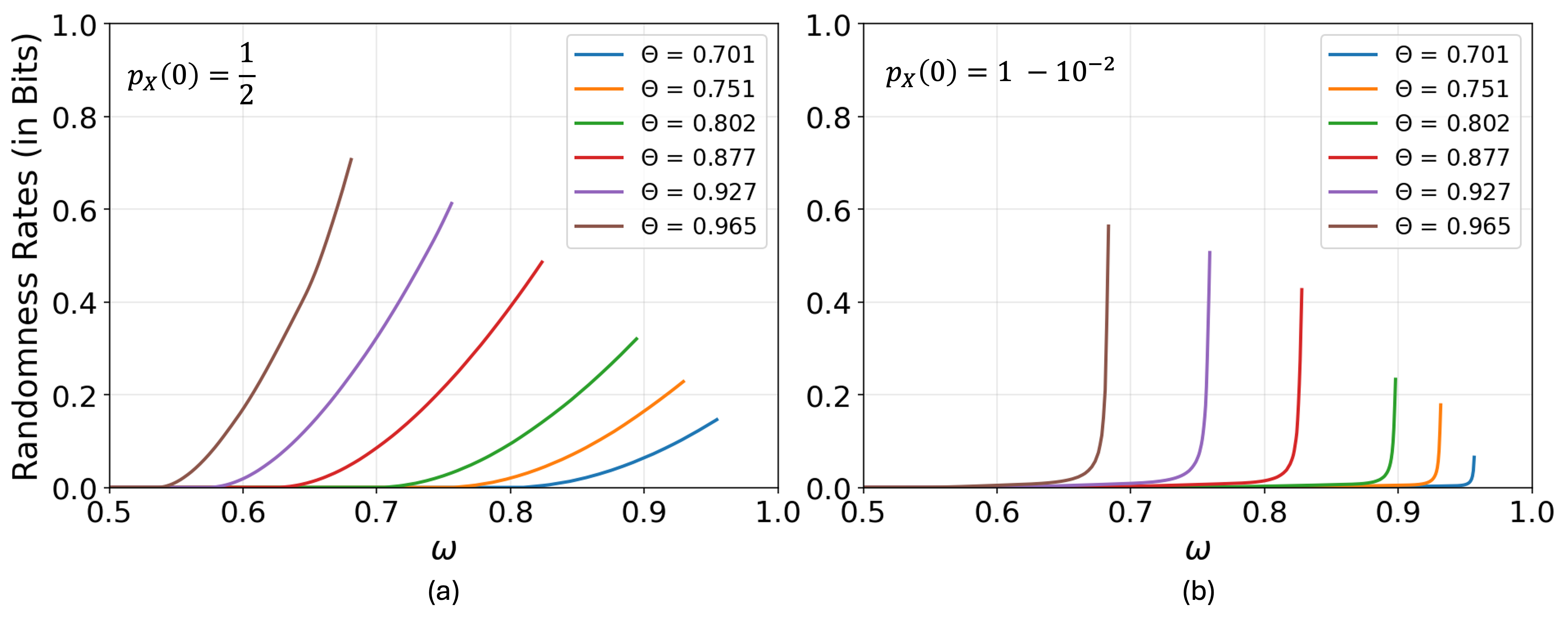}
    \caption{Asymptotic rates for the protocols as a function of the score $\score$ for different values of overlap $\overlap$. Figure (a) gives the rates for Protocol \ref{protocol: Semi-DI recycling} when $p_{X}(0) = 1/2$. Figure (b) gives the rates for Protocol \ref{protocol: private to public} from~\eqref{eqn: rand_expansion} when $p_{X}(0) = 1 - 1/100$. The testing probability $\gamma$ is taken to be approximately $0$ to compute the asymptotic rates.}
    \label{fig:asymptotic_rates_Semi_DI}
\end{figure}

In Figures~\ref{fig:asymptotic_rates_Semi_DI}(a) and~(b), we plot the rates for Protocols~\ref{protocol: Semi-DI recycling} and~\ref{protocol: private to public}, respectively, as a function of the score $\score$ for various values of overlap $\overlap$. For a fixed overlap, as the score $\score$ increases, the randomness rate increases until a maximum score is reached, beyond which there are no quantum strategies achieving the score $\score$ for that overlap.

As anticipated, for both protocols the randomness rate decreases with overlap for a fixed score, and, for smaller overlaps, higher scores become essential for generating randomness. Conversely, for large overlaps, only relatively small scores can be achieved, which can make large overlaps suboptimal for the protocol.

As a result, there is an optimal range of $\overlap$ values—not too large or too small—for which sufficiently high and experimentally achievable scores $\score$ are best suited for randomness generation. This range also depends on the input distribution $p_X$. A general observation is that for more biased input distributions, the suitable range of overlap for generating randomness is higher. This trend can be observed by comparing Figures~\ref{fig:asymptotic_rates_Semi_DI}(a) and~(b).

We have plotted the randomness for Protocol~\ref{protocol: private to public} for $p_{X}(0) =   1 -  1/100$ (See Figure~\ref{fig:asymptotic_rates_Semi_DI}(b)). As expected, this protocol provides a lower randomness rate compared to Protocol~\ref{protocol: Semi-DI recycling}. However, if the figure of merit is the ratio of output to input randomness, then this protocol performs better than Protocol~\ref{protocol: Semi-DI recycling}, as the input randomness is almost negligible here. Although the randomness rate is lower, Protocol~\ref{protocol: private to public} provides rates of about $0.1$ bits per round in the asymptotic limit for overlaps around $0.8$, which is also reasonable. Note that here $p_{X}(0)$ can be further reduced to get slightly better asymptotic rates for the protocol.

Our results, particularly concerning Protocol~\ref{protocol: Semi-DI recycling}, demonstrate strong performance of the protocols under realistic experimental conditions. Keeping realistic experimental settings in mind, a wide range of overlaps between $0.6$ and $0.9$ yield high rates of randomness expansion, as depicted in Figure~\ref{fig:asymptotic_rates_Semi_DI}. For instance, by employing the strategy $\overlap = 0.9$ and $\score \approx 0.8$, we achieve an estimated $0.537$ bits per round for Protocol~\ref{protocol: Semi-DI recycling}. In contrast, achieving similar randomness rates in the DI counterpart of the CHSH-based recycling protocol necessitates a CHSH score significantly higher than what can be achieved with state-of-the-art techniques~\cite{LLR&}. This shows that our protocol achieves competitive asymptotic rates while maintaining a high level of security: we allow for quantum side information, do not assume i.i.d.\ behaviour, and make no structural assumptions about the source or measurement devices beyond what is enforced in-protocol. 

The advantage becomes even more pronounced in the finite-size regime. Using the Entropy Accumulation Theorem (EAT), we show that positive expansion is achievable with as few as $10^5$ rounds under honest-device conditions (Fig.~\ref{fig:finite_rates}). Even with detection efficiencies in the range 0.6–0.8 (some of which are below the threshold needed for DI protocols with qubit pairs) Protocol~\ref{protocol: Semi-DI recycling} achieves positive expansion with about $10^6$ to $10^7$ rounds. Moreover, finite-size randomness rates remain high even under realistic imperfections, reaching about $0.1$ bits for modest scores that can be achievable in the finite size rates. These results indicate that composable semi-DI randomness expansion under weak assumptions is practical with hardware and short seed requirements, marking a significant step toward deployable semi-DI QRNGs.

\begin{figure}[t]
  \centering
  \includegraphics[width=\linewidth]{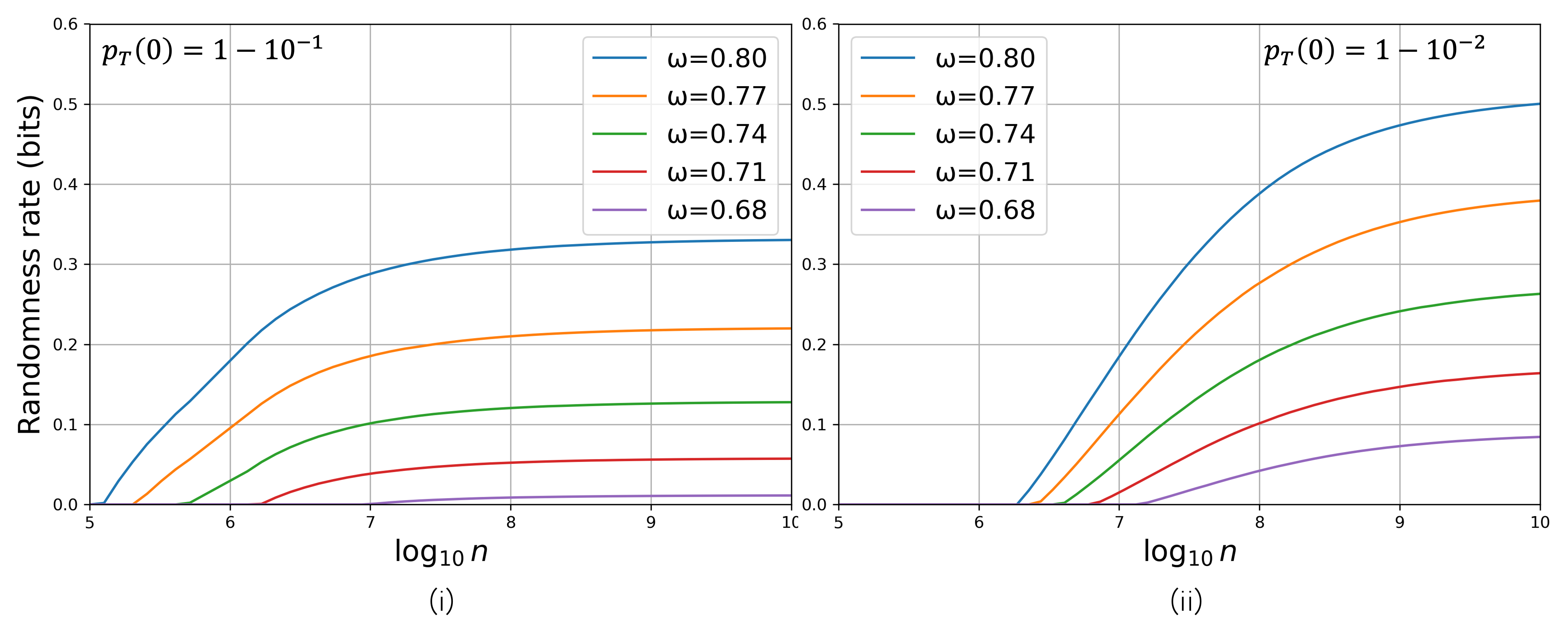}
  \caption{Finite-size randomness expansion rates for Protocol~\ref{protocol: Semi-DI recycling} as a function of the number of rounds $n$ (log scale) when $\overlap = 0.9$. Figure (i) corresponds to testing probability $\gamma=0.1$. Figure (ii) corresponds to testing probability $\gamma=0.01$. We have chosen the completeness error $\epsilon_C = 10^{-3}$ and soundness error to be $\epsilon_S = 10^{-6}$ to generate the curves.}
  \label{fig:finite_rates}
\end{figure}
The framework developed in the appendices enables us to compute rates by directly optimizing the von Neumann entropy. This is in contrast to much of the existing literature in semi-DI protocols, which typically involves lower bounding the von Neumann entropy in terms of the min-entropy and then optimizing the min-entropy (see, for example,~\cite{Tebyanian_2021}). Direct optimization of the von Neumann entropy results in higher randomness rates, even while making weaker security assumptions. For example, a protocol based on overlap discussed in Ref.~\cite{avesani2020} gives a maximum achievable rate using by lower bounding the single round min-entropy to $0.25$ bits per round, even under the stronger assumption that the source prepares two (unknown) coherent states and restricting to i.i.d.\ attacks. We note, however, that due to some simplifications carried out while computing rates, the bounds obtained by our method are not tight, and our rates could be improved further if a better technique is found. For instance, modifications of techniques such as~\cite{BFF} can potentially be used to find arbitrarily tight bounds on the von Neumann entropy for our protocol. This is left for future work.

\section{Discussion}\label{sec: discussion}  
In this work, we analyzed semi-DI randomness expansion protocols based on a prepare and measure setup. We proposed a protocol that recycles input randomness and another that achieves randomness expansion with heavily biased inputs. By using modest additional assumptions, our protocols are capable of generating higher randomness rates and are comparatively easier to implement than their DI counterparts.

It would be useful to extend this work to eliminate Assumption~\ref{ass: sDI-measurement}. More sophisticated mathematical techniques may be needed to do so, but this would be advantageous from the point of view of security.

Another avenue for extension is to broaden the scope of this work to be able to go beyond the case of $2$ inputs and $2$ outputs. The current proof relies on Jordan's lemma, limiting it to protocols with only $2$ inputs and $2$ outputs. Hence, a more general approach is desirable for other protocols. Recent advances in techniques for optimizing the von Neumann entropy~\cite{BFF2022} could potentially be leveraged to analyze these protocols for arbitrary numbers of inputs and outputs.

The structure and analysis of these protocols closely resemble DI protocols for randomness expansion, suggesting numerous opportunities for applying similar techniques and ideas. For example, it would be interesting to compute randomness rates conditioned on the full statistics $p(y|x)$ rather than a singular score.

Finally, it is worth noting that our protocols only require that the average overlap across both inputs surpasses a fixed value $\overlap$.  However, extending to the scenario where both values of individual overlaps for each input exceed a threshold value $\overlap$ could further enhance the rates.

\section{Additional Note} 
An earlier version of this work appeared in the PhD thesis of RB~\cite{Rutvij_thesis}. While writing this manuscript, a related work~\cite{pauwels2024prepare} appeared on arXiv, presenting a method to optimize Shannon entropies (i.e., under the assumption of a classical adversary). This work is also based on overlap constraints, though their protocols and security analysis differ significantly from ours.

\section{Acknowledgements} 

This work was supported by the UK Engineering and Physical Sciences Research Council (EPSRC) via the Quantum Communications Hub (Grant No.\ EP/T001011/1) and the Integrated Quantum Networks Hub (Grant No.\ EP/Z533208/1). RB thanks Stefan Weigert, Serge Massar and Lewis Wooltorton for their valuable comments on a previous version of this work.
\appendix

\section{Feasible Experimental Implementation}\label{sec: experimental}
In this appendix, we explore potential experimental implementations of our proposed protocol, focusing on various modulation, switching, and detection techniques suitable for different experimental conditions. We introduce possible approaches for state preparation, modulation, and detection, emphasizing that multiple alternative methods and setups can be employed based on specific experimental requirements. The key aspects of the protocol are:

\begin{itemize}
    \item \textbf{Source} $\mathbb{S}$: State preparation involves a laser source where modulation may employ an on-off keying (OOK) scheme. This method can be implemented using a mechanical shutter, an amplitude modulator, or by directly driving the laser to toggle between on and off states, effectively encoding information by either sending or not sending light. Such configurations generate either vacuum \(\ket{0}\) or coherent states \(\ket{\alpha}\). 
    \item \textbf{Switch} $\mathbb{W}$:  This unit determines whether it is a measurement or verification round by using Electro-Optic Modulators (EOMs) or Optical Switches (OS) to selectively direct the beam to a photodiode. The distribution of light between the measurement and verification rounds is determined by a RNG and controlled by adjusting the voltage applied to the modulation device.   
    \item \textbf{Measurement Device $\mathbb{M}$:} The quantum states are directed to the measurement apparatus, where detection is performed using a high-efficiency single-photon detector capable of identifying one or more photons or the absence of photons. The measurement results are then used to compute the input-output correlation and are further processed for post-processing and randomness extraction based on entropy values.
\end{itemize}

\begin{figure}[h!]\label{fig:EXP}
    \centering
    \includegraphics[width=\textwidth]{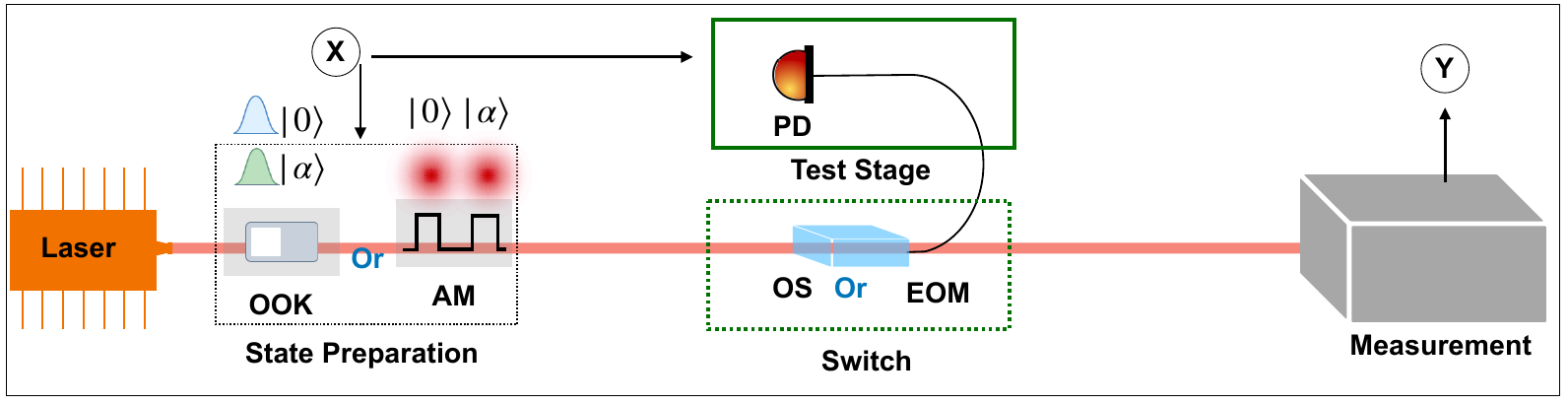}
    \caption{This figure illustrates the setup starting with a laser in the preparation stage where OOK or amplitude modulation generates the states \(\ket{0}\) or \(\ket{\alpha}\). A switch, using either EOMs or OS, determines the path of these states based on input, directing them to either test or measurement rounds. In the test round, a photodiode is used to measure whether any photons are present or not. The states are measured in the measurement stage, which is the final phase, and the measurement outcomes are utilized for further analysis.}
\end{figure}

\section{Security definition}\label{app: security definition}
For the security of the protocols, we use a composable security definition \cite{PUL&,LLR&}. Consider a protocol with output $Z$ and let $\Omega$ denote the event that it does not abort. The protocol is $(\epsilon_S,\epsilon_C)$-secure if
\begin{enumerate}
    \item $\frac{1}{2}p_{\Omega} \|\rho_{ZE|\Omega} - \frac{1}{d_Z}\mathbb{I}_Z \otimes \rho_{E|\Omega}\|_1 \leq \epsilon_S$,
    where $E$ is the quantum system held by the adversary, and $d_Z$ is the dimension of system $Z$; and
    \item There exists a quantum strategy such that $p_{\Omega}\geq1-\epsilon_C$.
\end{enumerate}
Here, \(\epsilon_S\) is the soundness error, and \(\epsilon_C\) is the completeness error. The completeness error is the probability that an honest protocol aborts, whereas the soundness error bounds how well the real protocol can be distinguished from an ideal protocol whose output is fully uncorrelated with any other system (including an adversary), and the outputs are uniformly distributed.

\section{Computing the randomness rate}\label{app: rates}
\subsection{Different systems}\label{app: EAT channels} 

In this work we use the notation $\mathcal{B}(\M{H})$ and $\mathcal{S}(\M{H})$ to be the set of bounded operators and density operators on a Hilbert space $\M{H}$ respectively. We now outline our notation for different registers: 
\begin{itemize}
    \item $X_{i}$: classical register storing the input of round $i$.
    \item $Y_{i}$: classical register storing the output of round $i$.
    \item $R_{i}$: register representing the quantum system stored in the source when round $i$ commences.
    \item $B_{i}$: quantum register denoting the system sent by the source to the measurement device (or the power meter). 
    \item $C_{i}$: quantum register held by the measurement device when round $i$ commences. This register may include pre-shared classical randomness with the source and the adversary. Furthermore, this may also include the information of the outcomes of previous rounds. 

    \item $E$: quantum register held by the adversary during the protocol.
\end{itemize}

It is assumed that the registers $C_i$ are not entangled with the registers $R_i$ and $E$; however, the $C_i$ registers may be (classically) correlated with $R_i$ and $E$. We note that the registers $R_i$ and $E$ can be entangled with each other (see Figure~\ref{fig: channel_diagram}).

\subsection{Channels of the protocol}
This section provides a detailed description of a single round of the protocol, as illustrated in Figure \ref{fig: channel_diagram}. Each round consists of three distinct types of channels: the preparation channel $\mathcal{P}_i$, the source channel $\mathcal{M}_i$, and the test channel $\mathcal{T}_i$. Here we describe the action of these channels on an initial state in detail.

\begin{figure}[!ht]\label{fig: channel_diagram}
\centering
\resizebox{0.7\textwidth}{!}{%
\begin{circuitikz}
\tikzstyle{every node}=[font=\LARGE]
\draw [short] (1.25,22) -- (13.75,22);
\draw [short] (1.25,22) -- (0,20.75);
\draw [short] (0,19.5) -- (2.5,19.5);
\draw [dashed] (0,18.25) -- (1.25,17);
\draw [dashed] (1.25,17) -- (8.75,17);
\draw  (8.75,18.25) rectangle (11.25,15.75);
\draw [dashed] (11.25,17) -- (13.75,17);
\draw  (2.5,20.75) rectangle (5,18.25);
\draw [short] (5,19.5) -- (13.75,19.5);
\draw [->] (3.75,18.25) -- (3.75,17.25);
\draw [->] (10,15.75) -- (10,14.5);
\draw [short] (5,19) -- (10,19);
\draw [short] (10,19) -- (10,18.25);
\node [font=\LARGE] at (10,17.00) {$\M{M}_{i}$};
\node [font=\LARGE] at (3.75,19.50) {$\M{P}_{i}$};
\node [font=\LARGE] at (1.6,17.30) {$C_{i}$};
\node [font=\LARGE] at (4.5,17.75) {$X_{i}$};
\node [font=\LARGE] at (7,19.80) {$R_{i+1}$};
\node [font=\LARGE] at (7,18.65) {$B_{i}$};
\node [font=\LARGE] at (10.5,15.25) {$Y_{i}$};
\node [font=\LARGE] at (1,19.80) {$R_{i}$};
\node [font=\LARGE] at (9.5,22.30) {$E$};
\node [font=\LARGE] at (4,21.25) {$\mathbb{S}$};
\node [font=\LARGE] at (11,18.75) {$\mathbb{M}$};
\node [font=\LARGE] at (12.5,17.30) {$C_{i+1}$};
\end{circuitikz}
}%

\caption{A diagram of the measurement round channel $\M{N}_{i}^{G} = (\cI_{ER_{i+1}} \ot \M{M}_{i}) \circ (\cI_{E} \ot \M{P}_{i} \ot \cI_{C}) $. The source is denoted by $\mathbb{S}$, and the measurement device is denoted by $\mathbb{M}$. The classical input and output registers are $X_{i}$ and $Y_{i}$, respectively. $X_{i}$ appears on the diagram as an output; its generation has been absorbed into the map $\M{P}_{i}$. }
\label{fig:semi_DI}
\end{figure}
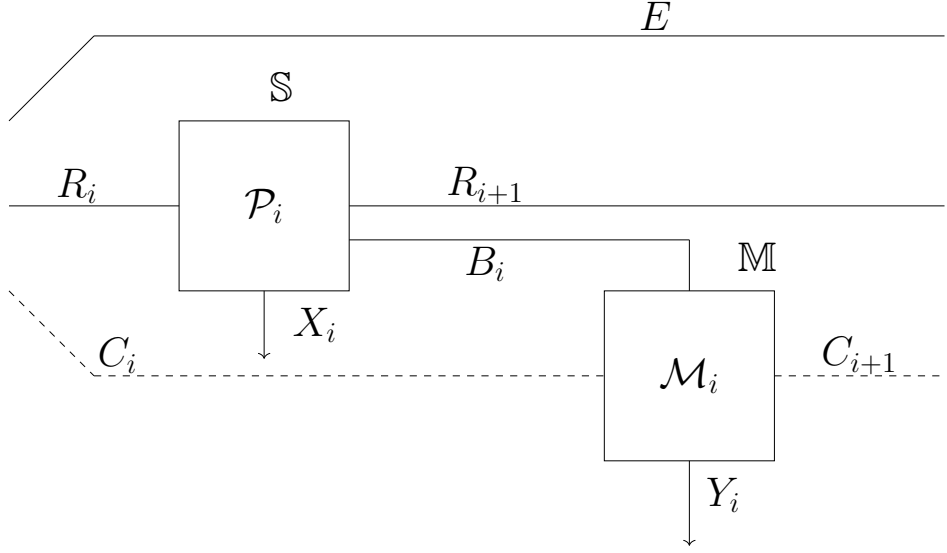

At the beginning of round $i$, the registers $R_{i}$ and $C_{i}$ can, in general, be classical correlated and the initial state including $E$ has the form 
\begin{eqnarray}\label{eqn: input state}
    \tau_{E R_i C_i} = \sum_{\lambda} p_{\Lambda}(\lambda) \tau_{E R_{i}}^{\lambda} \ot \tau_{C_i}^{\lambda} ,
\end{eqnarray}
where $\{ \tau_{E R_{i}}^{\lambda} \}_{\lambda}$ and $\{ \tau_{C_i}^{\lambda}\}_{\lambda} $ are some density operators. This captures the assumption that the register $C_{i}$ is not entangled with $R_i$ or $E$, but $R_i$ and $E$ may be entangled. We now describe the action of the source channel $\M{P}_i$, the measurement channel $\M{M}_i$, and the test channel $\M{T}_i$ in detail.

Note that we use a compressed notation: for two registers $A$ and $B$ we use $\M{Q}:A\to B$ to mean that $\M{Q}$ takes states (density operators) on $A$ to states on $B$.

\subsubsection{Source channel}
The action of the source can be described via the map $\M{P}_{i}: R_{i} \to X_{i} B_{i} R_{i+1}$, which takes the form 
\begin{eqnarray}\label{eqn: source channel}
    \M{P}_{i}(\tau_{R_i})  = \sum_{x} p_{X}(x) \proj{x}_{X_{i}} \ot \M{P}^{x}_{i}(\tau_{R_i}) ,
\end{eqnarray}
where $\{ \M{P}^{x}_{i} \}_{x}$ are channels $\M{P}^{x}_{i}: R_{i} \to  B_{i}{R}_{i+1} $. Here $R_{i+1} \equiv A_{i}$ is the system stored by the source and $B_{i}$ is the system sent to the measurement device.

If the input state is of the form \eqref{eqn: input state}, then the state after the action of the preparation channel is given by 
\begin{eqnarray}\label{eq:sourceout}
   \sigma_{X_iEB_iR_{i+1}C_i}=\cI_{E} \ot \M{P}_{i} \ot \cI_{C_{i}}(\sum_{\lambda} p(\lambda) \tau^{\lambda}_{E R_{i}} \ot \tau^{\lambda}_{C_{i}}) = \sum_{x,\lambda} p_{X}(x) \proj{x}_{X_i} \ot  p_\Lambda(\lambda) \sigma^{x,\lambda}_{EB_{i}R_{i+1}} \ot \tau^{\lambda}_{C_{i}},
\end{eqnarray}
where $\sigma^{x,\lambda}_{EB_{i}R_{i+1} } :=\cI_{E} \ot \M{P}^{x}(\tau^{\lambda}_{ER_{i } })$.\\

\subsubsection{Measurement channel}\label{app: measurement_channel}

 The measurement device can be described via $\M{M}_{i} : B_{i} C_{i} \to Y_{i} C_{i+1} $ and takes the form: 
\begin{eqnarray}\label{eqn: measurement channel}
        \M{M}_{i}(\sigma_{B_{i}C_{i}}  ) =  \sum_{y } \proj{y}_{Y_{i}} \ot \tilde{\M{E}}^{y}_{i}(\sigma_{B_{i} C_{i}} )   , 
\end{eqnarray}
where $\{\tilde{\M{E}}^0_i, \tilde{\M{E}}^1_i\}$ are sub-normalized channels (i.e., trace non-increasing completely positive maps) from $B_iC_i$ to $C_{i+1}$. As $\tilde{\M{E}}^y_i$ is a channel, it also has a Kraus operator representation. So the action of $\M{M}_i$ can be expressed as: 
\begin{eqnarray}
    \M{M}_i(\sigma_{B_iC_i}) = \sum_{y,k} \proj{y}_{Y_i} \ot \left(E_{i,y}^k\, \sigma_{B_iC_i} \left( E_{i,y}^k \right)^{\dagger}\right)
\end{eqnarray}
where $\{ E_{i,y}^k \}_{y,k} $ are some Kraus operators (each mapping $\M{H}_{B_iC_i}$ to $\M{H}_{C_{i+1}}$) satisfying $\sum_{y , k} (E_{i,y}^{k})^{\dagger} E_{i,y}^{k} = \id_{B_{i}C_{i}} $. If we take the state output by the source as given in~\eqref{eq:sourceout}, trace out $R_{i+1}$ and apply the measurement channel then we obtain
\begin{eqnarray}
\sum_{x,y,\lambda,k} p_X(x) \proj{x}_{X_i} \ot \proj{y}_{Y_i} \ot  p_\Lambda(\lambda) ( \id_E \ot E_{i,y}^{k})(\sigma^{x,\lambda}_{EB_{i}} \ot \tau^{\lambda}_{C_{i}})( \id_E \ot \left(E_{i,y}^k\right)^\dagger).
\end{eqnarray}
If we then trace out $C_{i+1}$ we get
\begin{eqnarray}
  &&\sum_{x,y,\lambda,k} p_X(x) \proj{x}_{X_i} \ot \proj{y}_{Y_i} \ot  p_\Lambda(\lambda)\underset{C_{i+1}}{\tr}[ ( \id_E \ot E_{i,y}^{k})(\sigma^{x,\lambda}_{EB_i} \ot \tau^{\lambda}_{C_i})( \id_E \ot \left(E_{i,y}^k\right)^\dagger)]\\
  &=&\sum_{x,y,\lambda,k} p_X(x) \proj{x}_{X_i} \ot \proj{y}_{Y_i} \ot  p_\Lambda(\lambda)\underset{B_iC_i}{\tr}[ ( \id_E \ot \left(E_{i,y}^k\right)^\dagger E_{i,y}^{k})(\sigma^{x,\lambda}_{EB_i} \ot \tau^{\lambda}_{C_i})]\\
  &=&\sum_{x,y,\lambda,k} p_X(x) \proj{x}_{X_i} \ot \proj{y}_{Y_i} \ot  p_\Lambda(\lambda)\underset{B_iC_i}{\tr}[(\id_E \ot (\id_{B_i} \ot \sqrt{\tau^{\lambda}_{C_i}})\left(E_{i,y}^k\right)^\dagger E_{i,y}^{k}(\id_{B_{i}} \ot \sqrt{\tau^{\lambda}_{C_i}}))\sigma^{x,\lambda}_{EB_{i}}]\\
  &=&\sum_{x,y,\lambda} p_X(x) \proj{x}_{X_i} \ot \proj{y}_{Y_i} \ot  p_\Lambda(\lambda)\underset{B_iC_i}{\tr}[(\id_E \ot M^{\lambda}_{i,y})\sigma^{x,\lambda}_{EB_i}],
\end{eqnarray}
where $M_{i,y}^{\lambda}=\displaystyle{\sum_{k}} \underset{C_i}{\tr} \left( \left( \id_{B_i} \ot  \sqrt{\tau^{\lambda}_{C_i}}\right)(E_{i,y}^{k})^{\dagger}E_{i,y}^{k}\left( \id_{B_i} \ot \sqrt{\tau^{\lambda}_{C_i}} \right) \right)$. This is the state we are interested in when considering one round of the protocol. By construction $M_{i,y}^{\lambda} \succeq 0$. Furthermore,
\begin{eqnarray}
\sum_{y} M_{i,y}^{\lambda} &=& \sum_{yk} \underset{C_{i}}{\tr} \left(\left(\id_{B_i} \ot \sqrt{\tau^{\lambda}_{C_i}}\right) (E^{k}_{i,y})^{\dagger} E^{k}_{i,y} \left(\id_{B_i} \ot \sqrt{\tau^{\lambda}_{C_i}}\right)  \right) \\ 
                               &=&  \underset{C_i}{\tr} \left(\left( \id_{B_i} \ot \sqrt{\tau^{\lambda}_{C_i}}\right) \sum_{k,y}\left( (E^{k}_{i, y})^{\dagger} E^{k}_{i, y} \right)\left(\id_{B_i} \ot \sqrt{\tau^{\lambda}_{C_i}}\right)  \right) \\ 
                               &=&  \underset{C_{i}}{\tr} \left( \id_{B_i} \ot \tau_{C_i}^{\lambda} \right) \\ 
                               &=& \id_{B_i}.
\end{eqnarray} 
Thus, $\{ M_{i,0}^{\lambda} , M_{i,1}^{\lambda} \}$ is a two-outcome POVM on the system $B_i$.

\subsubsection{Test channel} \label{app: test_channel}
The action of the test channel, applied in test rounds, is similar to the action of the measurement channel. The test channel $\M{T}_{i}: B_{i} \to Y_{i}$ is given by 
\begin{eqnarray}\label{eqn: test channel}
\M{T}(\sigma_{B_{i}}) = \tr(\Pi_0 \sigma_{B_{i}})  \proj{0}_{Y_{i}} +  \tr(( \id - \Pi_0) \sigma_{B_{i}})  \proj{1}_{Y_{i}}.
\end{eqnarray}

\subsubsection{Single round channels}\label{app: single_round_channels}
A single round of the protocol is described by the channel $\M{N}_{i}: R_{i} C_{i} \to T_{i} X_{i} Y_{i} R_{i+1} C_{i+1}$. Specifically, the channel $\M{N}_{i}$ takes the form: 
\begin{eqnarray}\label{eqn: single_round_channel}
   \M{N}_{i} =  \gamma \proj{1}_{T_{i}} \ot  \M{N}^{T}_{i} + (1 - \gamma) \proj{0}_{T_{i}} \ot  \M{N}^{G}_{i}.
\end{eqnarray} 
Here $\gamma$ is the testing probability. The channels $\M{N}^{T}_{i}$ and $\M{N}^{G}_{i}$ are of the following form: 
\begin{eqnarray}
    \M{N}^{T}_{i} &=& (\cI_{R_{i+1}} \ot \M{T}_{i} \ot \cI_{C_{i}} )\circ( \M{P}_{i}  \ot \cI_{C_{i}}) \\ 
    \M{N}^{G}_{i} &=& (\cI_{R_{i+1}} \ot \M{M}_{i} ) \circ (\M{P}_{i}  \ot \cI_{C_{i}} ), 
\end{eqnarray} 
where $\M{P}_{i}, \M{M}_{i} , \M{T}_{i}$ are given by equations \eqref{eqn: source channel}, \eqref{eqn: measurement channel} and \eqref{eqn: test channel}.

\subsection{Asymptotic randomness rate}\label{app: asymptotic rates}
The Entropy Accumulation Theorem~\cite{DFR,ARV}, shows that, in order to compute the asymptotic rates for the protocol, it suffices to focus only on a single representative round of the protocol.
The CQ state after one round of the protocol takes the following generic form: 
\begin{eqnarray}\label{eqn: CQ state for the protocol}
 \sum_{t,x,y} p(t) \proj{t} \otimes p_{X}(x) \proj{x} \otimes p_{Y|x,t}(y) \proj{y}_{Y} \otimes \rho^{t,x,y}_{E}.
\end{eqnarray}
Here, \(\rho^{t,x,y}_{E}\) are the states held by the adversary depending on $t$, $x$ and $y$. For the CQ state \eqref{eqn: CQ state for the protocol}, we make the following definitions. 
\begin{definition}[Overlap] Let $\rho_{TXYE}$ be a CQ state of the form \eqref{eqn: CQ state for the protocol}. Then the overlap $\overlapfunction(\rho_{TXYE})$ is given by: 
\begin{eqnarray}
    \overlapfunction(\rho_{TXYE}) = \frac{1}{2} \left( \sum_{x} p_{Y|x,t=1}(0)\right). 
\end{eqnarray}
\end{definition}
\begin{definition}[Score] Let $\rho_{TXYE}$ be a CQ state of the form \eqref{eqn: CQ state for the protocol}. Then the score $\scorefunction(\rho_{TXYE})$ is given by: 
\begin{eqnarray}
    \scorefunction(\rho_{TXYE}) = \frac{1}{2} \left( \sum_{x} p_{Y|x,t=0}(x)\right). 
\end{eqnarray}
\end{definition}
Note that $\overlapfunction(\cdot)$ and $\scorefunction(\cdot)$ refer to the functions on the CQ states that determine the overlap and the score obtained for a given CQ state, whereas $\overlap$ and $\score$ are used to refer to the actual values of overlap and score, respectively. For example, we can say that $\overlapfunction(\rho_{TXYE}) = \overlap$ and $\scorefunction(\rho_{TXYE}) = \score$.

We define the set $\Gamma[\score , \overlap]$ to be the set of all single-round strategies that achieve overlap $\overlap$ and score $\score$, i.e.,
\begin{eqnarray*}
\Gamma[\score , \overlap] := \{ \rho_{TXYE} =\!\! \underset{R_{i+1}C_{i+1}}{\tr}\!\!\left[\M{N}_{i}(\tau_{ERC})\right] :   \M{N}_{i} \text{ of form  \eqref{eqn: single_round_channel}},  \tau_{ERC} \text{ of form \eqref{eqn: input state}} , \scorefunction(\rho_{TXYE}) = \score \text{ and } \overlapfunction(\rho_{TXYE}) = \overlap \}.
\end{eqnarray*}

 The asymptotic randomness generation rate (randomness generation per round in the limit as the number of rounds tends to infinity) for the protocol is given by
\begin{eqnarray}\label{eqn: definition of F}
    \inf_{\rho_{TXYE} \in \Gamma[\score, \overlap]}   \bar{H}_{\rho_{TXYE}},
\end{eqnarray} 
where $\bar{H}$ is either $H(XY|E)$ when referring to Protocol~\ref{protocol: Semi-DI recycling} or $H(Y|XE)$ when referring to Protocol~\ref{protocol: private to public}. While we do not compute the rates for finite rounds in this work, the EAT can be used for this purpose as well, by incorporating a penalty that scales as $1/\sqrt{n}$.

In this work we will only extract the randomness from the outputs generated generated during the measurement rounds (in general, there may be some randomness in the test rounds as well, but in the present paper, we do not extract this randomness). Furthermore, in our protocol, if the probability of a test round is small and hence so is the expected number of test rounds. Hence ignoring test rounds when extracting has minimal effect on the randomness.

When the randomness is only extracted from the rounds with $T = 0$, the randomness generation rate\footnote{Note that in the case when the variable $T$ is unknown to the adversary, we can also take the randomness generation rate to be $H(\cdot|\cdot E)$. In this case too the entropic quantity $p_{T}(0)H(\cdot|\cdot,T=0,E)$ is a valid lower bound on the randomness generation rate because $H(\cdot|\cdot E) \geq H(\cdot|\cdot TE) \geq p_{T}(0) H(\cdot|\cdot T=0,E).$} can be computed as:
\begin{eqnarray} 
\bar{H}_{\rho} = (1 - \gamma) \hat{H}_{\rho}, \nonumber 
\end{eqnarray}
where $\gamma$ is the testing probability and $\hat{H} = H(XY|T=0,E)$ when Protocol~\ref{protocol: Semi-DI recycling} is considered and $\hat{H} = H(Y|X,T=0,E)$ when Protocol~\ref{protocol: private to public} is considered.

Randomness is consumed in two places in Protocol~\ref{protocol: Semi-DI recycling}: when choosing the input $X$, and when choosing whether or not to test. The input randomness rate is hence $H(X) + H(T)$. In Protocol~\ref{protocol: private to public} a public source of randomness is used to choose $X$ and $T$, and hence we do not include a penalty here.

We are interested in the randomness expansion rate, which is defined as the difference between the randomness generation rate and the randomness consumption rate of the protocol. For Protocol~\ref{protocol: Semi-DI recycling}, the expression for randomness expansion rate can be computed using the chain rule of von-Neumann entropy (see as in Protocol $2$ in~\cite{BhRC})

\begin{eqnarray*}
r_{\ref{protocol: Semi-DI recycling}}\! = \inf_{\rho \in \Gamma[\score, \overlap]} \left(H(XTY|E)_{\rho} \right) - H(XT) = \inf_{\rho \in \Gamma[\score, \overlap
]}H(Y|TXE)_{\rho} \geq (1 - \gamma) \left(\inf_{\rho \in \Gamma[\score , \overlap]} H(Y|X, T=0, E) \right)
\end{eqnarray*}

In the case when the randomness from the string $\B{T}$ is not recycled, the rate is modified to

\begin{eqnarray*}
r_{\ref{protocol: Semi-DI recycling}}'\! =\! (1\! -\! \gamma)\!\left(\inf_{\rho \in \Gamma[\score , \overlap]} H(XY|T=0,E)_{\rho} \right)\! -\! H(X)\! -\! H(T)\! =\! (1\! -\! \gamma)\!\left(\inf_{\rho \in \Gamma[\score , \overlap]} H(Y|X,T=0,E)_{\rho} \right)\!-\! \gamma H_{\bin}(p_{X}(0))\! -\! H_{\bin}(\gamma) . 
\end{eqnarray*}
For Protocol~\ref{protocol: private to public} we can compute the rate as 
\begin{eqnarray*}
r_{\ref{protocol: private to public}} = (1 - \gamma)\left(\inf_{\rho \in \Gamma[\score , \overlap]} H(Y|X,T=0,E)_{\rho} \right)
\end{eqnarray*}

From the discussion above, the randomness expansion rates for both protocols depend on
\begin{eqnarray}\label{eqn: F_function}
    F_{p_{X}}(\score , \overlap) := \inf_{\rho \in \Gamma[\score , \overlap]} H(Y | X ,T=0 , E)_{\rho},
\end{eqnarray}
and the remainder of this appendix is devoted to computing lower bounds on $F_{p_{X}}(\score , \overlap)$.

\subsection{Reducing the strategy space}\label{app: reducing the strategy space}
Our first step for solving the optimization problem \eqref{eqn: F_function} is to simplify the set over which the entropy $H(Y|X, T = 0, E)$ is optimized. The goal of this section is to show that, in order to compute the rate, it suffices to consider strategies where the adversary holds a purification of the state prepared by the source, and where the measurement device performs a projective measurement instead of arbitrary POVMs. We begin by formalizing this.

If the input state $\tau_{E R_i C_i}$ is of the form~\eqref{eqn: input state}, then, from Sections~\ref{app: measurement_channel}, \ref{app: test_channel}, and \ref{app: single_round_channels}, the state $\underset{R_{i+1}C_{i+1}}{\tr}\M{N}_i(\tau_{E R_i C_i})$ can be expressed in the form 
\begin{eqnarray}\label{eqn: CQ state explicit}
\rho_{TX_iY_iE} = \sum_{t, x, y} p_T(t) p_X(x) \proj{t} \otimes \proj{x}_{X_{i}} \otimes \proj{y}_{Y_{i}} \otimes \sum_{\lambda} p_\Lambda(\lambda) \tr_{B_i} \left(\left(M_y^{t, \lambda} \otimes \id_E\right) \sigma^{x, \lambda}_{B_iE}\right),
\end{eqnarray}  
where, $M^{1, \lambda}_0 \equiv M^1_0 = \Pi_0$ and $M^{1, \lambda}_1 = M^1_\lambda = \id_{B_i} - \Pi_0$ are used to obtain the compact form above. Note that the random variables $\Lambda$ and $X$ are independent in the CQ state above.

Since we are considering single-round strategies we henceforth drop the subscript $i$ on the variables and define the following subsets of \(\Gamma[\score, \overlap]\):
\begin{itemize}
    \item \(\Gamma_{\mathrm{Proj}}[\score, \overlap]\) is the subset of all CQ states of the form \eqref{eqn: CQ state explicit} such that \(M^{0,\lambda}_{y}\) are projectors for every \(\lambda\) and \(y\).
    \item \(\Gamma_{\mathrm{Pure}}[\score, \overlap]\) is the subset of all CQ states of the form \eqref{eqn: CQ state explicit} such that \(\sigma^{x, \lambda}_{BE}\) are pure states.
    \item $\Gamma_{s}[\score, \overlap]$ is the subset of all CQ states of the form:
    \begin{eqnarray}\label{eqn: single CQ state}
        \rho_{TXYE} = \sum_{t, x, y} p_T(t) p_{X}(x) \proj{t} \otimes \proj{x} \otimes \proj{y} \otimes \tr_B\left(\left(M_y^t\otimes\id_E\right) \sigma^x_{BE}\right),
    \end{eqnarray} 
    where $\sigma^{x}_{BE}$ are pure states and $M_{y}^{t}$ are projectors. Note that the state in \eqref{eqn: single CQ state} does not contain any dependence on (classical) variable $\lambda$ as opposed to a general state given by \eqref{eqn: CQ state explicit}.
\end{itemize}
The goal of this section is to simplify the optimization problem~\eqref{eqn: F_function}. In particular, we show that to solve~\eqref{eqn: F_function}, it is sufficient to replace the set $\Gamma[\score, \overlap]$ being optimized over with the set $\Gamma_{s}[\score, \overlap]$.
\begin{lemma}\label{lemm: reduction to pure states}  
\begin{eqnarray}
    \inf_{\rho  \in \Gamma[\score, \overlap]} H(Y|X,T=0,E)_{\rho} = \inf_{\rho \in \Gamma_{\mathrm{Pure}}[\score, \overlap]} H(Y|X,T=0,E)_{\rho}.
\end{eqnarray} 
\end{lemma}
\begin{proof} 
Let $\rho_{TXYE}$ be any state of the form \eqref{eqn: CQ state explicit}. Let $\sigma_{BEE'}^{x,\lambda}$ be any purification of the $\sigma_{BE}^{x,\lambda}$. The state 
\begin{eqnarray}
    \bar{\rho}_{TXYEE'} = \sum_{t , x ,y}  p_{T}(t) p_{X}(x) \proj{t} \otimes \proj{x} \otimes \proj{y} \otimes \sum_{\lambda} p_{\Lambda}(\lambda) \tr_{B} \left(\left(M_{y}^{t, \lambda} \otimes \id_{E E'}\right) \sigma^{x, \lambda}_{BEE'}\right)
\end{eqnarray}
satisfies $\tr_{E'} \bar{\rho}_{TXYEE'}  = \rho_{TXYE}$, and hence $\scorefunction(\bar{\rho}) = \scorefunction(\rho)$ and $\overlapfunction(\bar{\rho}) = \overlapfunction(\rho)$. Hence, $\bar{\rho} \in \Gamma[\score, \overlap]$. Using the strong sub-additivity of the conditional von-Neumann entropy, we have 
\begin{eqnarray}
    H(Y|X,T=0,EE')_{\bar{\rho}} \leq  H(Y|X,T=0, E)_{\rho}.
\end{eqnarray}
Hence, for any state $\rho_{TXYE}\in\Gamma[\score,\overlap]$ the state $\bar{\rho}_{TXYEE'}\in\Gamma_{\mathrm{Pure}}[\score,\overlap]$ cannot lead to a higher value of $H(Y|X,T=0,E)$, so we can restrict the infimum to the set $\Gamma_{\mathrm{Pure}}[\score,\overlap]\subset\Gamma[\score,\overlap]$ without changing the result.
\end{proof}

We now reduce the analysis to the strategies in which the measurement device uses projective measurements.

\begin{lemma}\label{lemm: reduction to projective measurements}
The sets $\Gamma[\score , \overlap]$ and $\Gamma_{\mathrm{Proj}}[\score, \overlap]$ are identical.
\end{lemma}
\begin{proof} 
Let $\rho$ be of the form \eqref{eqn: CQ state explicit} with $\scorefunction(\rho) = \score$ and $\overlapfunction(\rho) = \overlap$, i.e., $\rho \in \Gamma(\score ,\overlap)$.  Let $M^{0,\lambda}_{0}$ and $M^{1 , \lambda}_{1}$ be non-extremal effects, 
so that we can express 
\begin{eqnarray}
    M^{0,\lambda}_{0} = \sum_{a} p_{A|\lambda}(a) P_{0}^{(a , \lambda)} 
\end{eqnarray}
where $P_{0}^{(a , \lambda)}$ are projectors on $\M{H}_{B}$ (which could include the zero projector) and $p_{A|\lambda}$ is a probability distribution. Furthermore, defining the projectors $P_{1}^{(a, \lambda)} = \id_{B} - P_{0}^{(a, \lambda)}$ allows us to express $M_{1}^{0,\lambda}$ as a convex sum:
\begin{eqnarray}
M^{0, \lambda}_{1} &=& \id_{B} -   M^{0, \lambda}_{0} \nonumber \\ 
                    &=& \sum_{a} p_{A|\lambda}(a) (\id_{B} -  P^{(a, \lambda) }_{0} ) \nonumber \\ 
                    &=& \sum_{a} p_{A|\lambda}(a) P^{(a, \lambda)}_{1}.
\end{eqnarray}
To prove the lemma, it suffices to show that $\rho$ admits an alternative decomposition of the form \eqref{eqn: CQ state explicit} only in terms of projectors. We have
\begin{eqnarray}
\rho_{TXYE} &=& \sum_{x, y} p_{X}(x) \Bigg( p_{T}(0)  \proj{0,x,y} \otimes \sum_{\lambda , a} \underset{B}{\tr}\left( p_{A|\lambda}(a) P_{\Lambda}(\lambda)  (P_y^{(a , \lambda)} \otimes \id_{E}) \sigma^{x,\lambda}_{BE}\right)  \nonumber\\ 
& &\hspace{7cm} + p_{T}(1) \proj{1 ,x ,y} \otimes \sum_{\lambda } \underset{B}{\tr}\left( p_{\Lambda}(\lambda)  (M^{1}_{y} \otimes \id_{E}) \sigma^{x,\lambda}_{BE}\right) \Bigg) \nonumber \\ 
&=& \sum_{x, y} p_{X}(x) \Bigg( p_{T}(0)  \proj{0 ,x ,y} \otimes \sum_{\lambda , a} \underset{B}{\tr}\left( p_{A,\Lambda}(a , \lambda) (P_y^{(a, \lambda)} \otimes \id_{E}) \sigma^{x,\lambda}_{BE}\right)  \nonumber \\ 
& & \hspace{7cm} + p_{T}(1) \proj{1 ,x ,y} \otimes \sum_{\lambda , a} \underset{B}{\tr}\left( p_{A , \Lambda}(a , \lambda) (M^{1}_{y} \otimes \id_{E}) \sigma^{x,\lambda}_{BE}\right) \Bigg) \nonumber \\
&=& \sum_{t, x, y} p_{X}(x) \Bigg( p_{T}(0) \proj{0 , x , y} \otimes \sum_{\lambda'} \underset{B}{\tr}\left( {p}_{\Lambda'}(\lambda' ) (P_y^{\lambda'} \otimes \id_{E}) \sigma^{x,\lambda'}_{BE}\right)  \nonumber\\ 
& & \hspace{7cm}+ p_{T}(1) \proj{1 , x , y} \ot \sum_{\lambda'} \underset{B}{\tr}\left( p_{\Lambda'}(\lambda')  (M^{1}_{y} \otimes \id_{E}) \sigma^{x,\lambda'}_{BE}\right)  \Bigg), \nonumber  
\end{eqnarray}
where in the second step, we have defined a joint probability distribution $p_{A,\Lambda}$ via $p_{A,\Lambda}(a, \lambda) =  p_{A|\lambda}(a) p_{\Lambda}(\lambda)$, and in the final step, we have defined $\Lambda' = (A, \Lambda)$ so that $p_{\Lambda'} = p_{A,\Lambda}$, and set $\sigma^{x, \lambda'}_{BE} \equiv \sigma^{(x, a,\lambda)}_{BE} = \sigma^{x,\lambda}_{BE}$. Note here that the random variable $\Lambda'$ is independent of the input $X$.

We have hence shown that $\rho_{TXYE} \in \Gamma_{\mathrm{Proj}}(\score , \overlap)$.
\end{proof} 

Finally we prove the main result of this section, which allows the expansion rate to be expressed in terms of an optimization over $\Gamma_s[\score,\overlap]$ and the convex envelope of a function, denoted $\mathrm{convenv}(\cdot)$\footnote{The convex envelope is also referred to as the ``convex floor'' or ``convex lower bound'' in the literature. See Appendix \ref{app: LF transform} for the definition.}.

\begin{lemma}\label{lemm: F_to_G}
  Let
  \begin{eqnarray}\label{def:G}
    G_{p_{X}}(\score, \overlap) =  \inf_{\rho \in \Gamma_{s}[\score, \overlap]} H(Y|X,T= 0,E)_{\rho}.
  \end{eqnarray}
  Then
\begin{eqnarray}
    F_{p_{X}}(\score, \overlap) = \mathrm{convenv}(G_{p_{X}})(\score, \overlap). 
\end{eqnarray}
\end{lemma}
\begin{proof} 
We prove the result in two steps, establishing (a) $F_{p_{X}}(\score, \overlap) \geq \text{convenv}(G_{p_{X}})(\score, \overlap)$, and then (b) $F_{p_{X}}(\score, \overlap) \leq \text{convenv}(G_{p_{X}})(\score, \overlap)$.\medskip

\noindent\textbf{Proof of (a):}

\noindent Consider any $\rho_{TXYE} \in \Gamma[\score, \overlap]$ of the form \eqref{eqn: CQ state explicit}. Lemma~\ref{lemm: reduction to pure states} implies that we can consider $\rho \in \Gamma_{\mathrm{Pure}}[\score, \overlap]$ when doing the optimization required for $F_{P_X}$ (cf.~\eqref{eqn: F_function}). Using Lemma~\ref{lemm: reduction to projective measurements}, we express $\rho_{TXYE} = \sum_{\lambda} p_{\Lambda}(\lambda) \rho^{\lambda}_{TXYE}$, where
\begin{equation*}
\rho^{\lambda}_{TXYE} :=  \sum_{t, x, y} p_{T}(t) p_{X}(x) \proj{t} \otimes \proj{x} \otimes \proj{y} \otimes \underset{B}{\tr} \left((M_{y}^{t, \lambda} \otimes \id_{E}) \sigma^{x, \lambda}_{BE}\right).
\end{equation*}
and $M_{y}^{t , \lambda}$ are projectors. Thus, $\rho^{\lambda}_{TXYE}\in \Gamma_{s}[\scorefunction(\rho^{\lambda}),\overlapfunction(\rho^{\lambda}))]$ by definition of $\Gamma_{s}$. Now, it is straightforward to see that
\begin{equation}\label{eqn: decomposition of probabilities}
\forall\ t,x,y \quad p_{Y|t,x}(y) = \tr \left( \sum_{\lambda} p_{\Lambda}(\lambda) (M^{t,\lambda}_{y} \otimes \id_{E}) \sigma^{x}_{BE}  \right) = \sum_{\lambda} p_{\Lambda}(\lambda)  \tr \left((M^{t,\lambda}_{y} \otimes \id_{E}) \sigma^{x}_{BE} \right),
\end{equation}
which implies that $
    \scorefunction(\rho_{TXYE}) = \sum_{\lambda} p_{\Lambda}(\lambda) \scorefunction(\rho^{\lambda}_{TXYE} ) $ and $
    \overlapfunction(\rho) = \sum_{\lambda} p_{\Lambda}(\lambda) \overlapfunction(\rho^{\lambda}_{TXYE}) $ (i.e., $\scorefunction(\cdot)$ and $\overlapfunction(\cdot)$ are linear maps).

    Using the concavity of the conditional von-Neumann entropy we have 
\begin{eqnarray*}
    H(Y|X, T=0, E)_{\sum_{\lambda} p_{\Lambda}(\lambda) \rho^{\lambda}_{TXYE}} &\geq& \underset{\lambda}{\sum} p_{\Lambda}(\lambda) H(Y|X,T=0,E)_{\rho^{\lambda}_{TXYE}} \\ 
    &\geq&  \sum_{\lambda} p_{\Lambda}(\lambda) \inf_{\rho'_{TXYE} \in \Gamma_{s}[\scorefunction(\rho^{\lambda}_{TXYE}), \overlapfunction(\rho^{\lambda}_{TXYE})]} H(Y|X,T=0,E)_{\rho_{TXYE}'} \\  
&=& \sum_{\lambda} p_{\Lambda}(\lambda) G_{p_{X}}(\scorefunction(\rho^{\lambda}_{TXYE}) , \overlapfunction(\rho^{\lambda}_{TXYE}))\\
&\geq&\mathrm{convenv}(G_{p_X})(\score, \overlap).
\end{eqnarray*}
In the last line we have used that the distribution $p_{\Lambda}$ induces a probability measure $\mu$ over points $(\score_{\lambda}, \overlap_{\lambda})$, satisfying
\[
\int (\score_{\lambda}, \overlap_{\lambda}) \, \mathrm{d}\mu = (\score, \overlap)
\]
and that the quantity $\sum_{\lambda} p_{\Lambda}(\lambda) G_{p_X}(\score_{\lambda}, \overlap_{\lambda})$ corresponds to
\[
\int f(\mathbf{x}') \, \mathrm{d}\mu(\mathbf{x}')
\]
in the definition of the convex envelope~\eqref{def:convenv}. Taking the infimum over all such distributions gives the final bound.\medskip

\noindent \textbf{Proof of (b):}\\
Consider an arbitrary convex combination $(\score,\overlap)=\sum_\lambda p_{\Lambda}(\lambda)(\score_\lambda,\overlap_\lambda)$ and suppose $\rho^{*\lambda}_{TXYE}$ is the optimal solution to the optimization problem $\underset{\rho_{TXYE} \in \Gamma_{s}[\score_{\lambda}, \overlap_{\lambda}]}{\inf} H(Y|X ,T=0,E)_{\rho_{TXYE}}$, i.e., $G_{p_X}(\score_{\lambda}, \overlap_{\lambda})=H(Y|X,T=0,E)_{\rho^{*\lambda}_{TXYE}}$. From the definition of $\Gamma_s[\score_{\lambda}, \overlap_{\lambda}]$, $\rho^{*\lambda}_{TXYE}$ takes the form 
\begin{eqnarray}
 \sum_{t, x, y} p_{T}(t) p_{X}(x) \proj{t} \otimes \proj{x} \otimes \proj{y} \otimes \tr_{B} \left((M_{y}^{*t,\lambda} \otimes \id_{E})\sigma^{* x,\lambda}_{BE}\right).   
\end{eqnarray}
Now consider a state of the form
\begin{eqnarray}
\rho_{TXYE} = \sum_{t, x, y} p_T(t) p_X(x) \proj{t} \otimes \proj{x} \otimes \proj{y} \otimes \sum_{\lambda} p_{\Lambda}(\lambda) \tr_B \left((M_y^{*t,\lambda} \otimes \id_E)\sigma^{*x,\lambda}_{BE_{\lambda}} \right),
\end{eqnarray}
which has been constructed in such a way that $E \simeq \bigoplus_{\lambda} E_{\lambda}$ (i.e., $E_{\lambda}$ and $E_{\lambda'}$ have orthogonal subspaces unless $\lambda \neq \lambda'$). From the linearity of $\scorefunction$ and $\overlapfunction$ we conclude that $\rho_{TXYE} \in \Gamma[\sum_{\lambda} p_{\Lambda}(\lambda)\score_{\lambda} , \sum_{\lambda} p_{\Lambda}(\lambda)  \overlap_{\lambda}]$. It follows that
\begin{eqnarray}
    H(Y|X,T=0,E)_{\rho_{TXYE}} \geq \inf_{\rho'_{TXYE} \in \Gamma[\score,\overlap]} H(Y|X,T=0,E)_{\rho'_{TYXE}}=F_{p_X}(\score,\overlap).
\end{eqnarray}
However, since $\tr_{B} ((M^{*t,\lambda}_{y} \ot \id_{E}) \sigma^{*x,\lambda}_{BE_{\lambda}})$ and $\tr_{B} ((M^{*t,\lambda'}_{y} \ot \id_{E}) \sigma^{*x,\lambda'}_{BE_{\lambda'}}) $ have orthogonal supports whenever $\lambda\neq\lambda'$, we must have that 
\begin{eqnarray}
    F_{p_X}(\score,\overlap)&\leq&H(Y|X,T=0,E)_{\rho_{TXYE}}\\
    &=&\sum_{\lambda} p_{\Lambda}(\lambda) H(Y|X,T=0,E)_{\rho^{*\lambda}_{TXYE}}\nonumber\\
    &=&\sum_{\lambda} p_{\Lambda}(\lambda)G_{p_X}(\score_\lambda,\overlap_\lambda).\label{eq:Fsum}
\end{eqnarray}
Thus, for any probability measure $\mu$ whose support is the points $\{(\score_{\lambda}, \overlap_{\lambda})\}_\lambda$, we can realize a state whose entropy equals $\int G_{p_X}(\score,\overlap)\,\mathrm{d}\mu(\score,\overlap)$.\footnote{Equation~\eqref{eq:Fsum} has a sum rather than an integral. Carathéodory's theorem states that any point in the convex hull of $A\subset\mathbb{R}^d$ can be expressed as the convex combination of $d+1$ points. If we take $A:=\{\, (\mathbf{x},\, G_{p_X}(\mathbf{x})) : \mathbf{x} = (\score_{\lambda},\overlap_{\lambda}) \,\}\subset \mathbb{R}^3$, then any point in the convex hull of $A$ can be expressed as a convex combination of at most $4$ points. Therefore the optimisation can be restricted to probability mass functions. In particular, for any $\mathbf{x}$ we have 
$(\mathbf{x},\mathrm{convenv}(G_{p_X})(\mathbf{x})) \in \mathrm{conv}(A)$, so the convex envelope can always be realised using a finite distribution.}

The construction above shows that for any probability measure $\mu$ whose support is the points $\{(\score_{\lambda}, \overlap_{\lambda})\}_\lambda$, we can realize a state whose entropy equals $\int G_{p_X}(\mathbf{x}') \, \mathrm{d}\mu(\mathbf{x}')$. As above, we also have that $\int (\score_{\lambda}, \overlap_{\lambda}) \, \mathrm{d}\mu = (\score , \overlap)$. By the definition of the convex envelope,
\[
\mathrm{convenv}(G_{p_X})(\score , \overlap) = \inf_{\mu} \Big\{ \int G_{p_X}(\mathbf{x}') \, \mathrm{d}\mu(\mathbf{x}') : \int \mathbf{x}' \, \mathrm{d}\mu = (\score,\overlap) \Big\}.
\]
Thus, we can always choose the distribution $\mu$ such that $\int G_{p_X}(\mathbf{x}') \, \mathrm{d}\mu(\mathbf{x}') = \mathrm{convenv}(G_{p_X})(\score , \overlap)$. Therefore, $\mathrm{convenv}(G_{p_X})(\score , \overlap) \geq F_{p_X}(\score , \overlap)$.

\end{proof}

Note that from the lemma above, $G_{p_{X}}(\score , \overlap) = 0 $ implies $F_{p_{X}}(\score , \overlap) = 0 $. Furthermore, one can use $ G_{p_{X}} $ values to help find a set of parameters $(\score , \overlap)$ that provides a good amount of randomness expansion without having to compute its convex envelope \( F_{p_{X}} \).

\subsection{Reduction to qubit strategies}\label{app: Jordans lemma}
The main issue with the optimization problem in the definition of $G_{p_{X}}(\score, \overlap)$ (cf.~\eqref{def:G}) is that $\M{H}_{B}$ could have an arbitrarily large dimension. The aim of this section is to reduce the optimization strategy space of the problem to strategies where the source prepares a qubit state and the measurement device performs projective measurements. To achieve this, we use Jordan's Lemma, to relate the problem to an analogous one for qubits. Then, we explicitly compute the score, overlap, and entropy for qubit strategies. Finally, based on these results, we relax the optimization problem for $G_{p_{X}}(\score, \overlap)$ to another optimization function $\tilde{G}_{p_{X}}(\score, \overlap)$, where the dimension of the states prepared by the source is restricted to two.

\subsubsection{Jordan's Lemma} 
Jordan's lemma~\cite{Jordan} allows the problem to effectively reduced to qubits. 
\begin{lemma}[Jordan's Lemma]Let $A, B \in B(\mathcal{H})$ be two projectors. Then 
\begin{equation}
    \mathcal{H}  = \bigoplus_{\alpha} \mathcal{H}_{\alpha}  
\end{equation}
such that each $\mathcal{H}_{\alpha}$ is an invariant subspace of $\mathcal{H}$ under the action of $A$, $B$, $\id - A $ and $\id - B$ . Moreover, the dimension of each subspace $\mathcal{H}_{\alpha}$ is at most 2. 
\end{lemma}

We now prove the following elementary result that will be useful later on. 
\begin{lemma}\label{lem: linear-algebra}
Let $\M{H}$ be a Hilbert space that admits the decomposition $\M{H} = \bigoplus_{\alpha} \M{H}_{\alpha}$, where each subspace $\M{H}_{\alpha}$ is invariant under the action of the operator $\hat{O}$. If $Q_{\alpha}$ is the projection onto the sub-space $\M{H}_{\alpha}$ then 
\begin{equation}
\forall \alpha:  [Q_{\alpha} , \hat{O}] = 0
\end{equation}
\end{lemma}
\begin{proof}
To prove the lemma, we show that $\forall \alpha,\ \forall \ket{\psi}\in\M{H}: Q_{\alpha} \hat{O} \ket{\psi} = Q_{\alpha} \hat{O} \ket{\psi}$. First, notice that we can express $\ket{\psi}$ as a linear combination of vectors $\ket{v_{\alpha'}} \in \M{H}_{\alpha'}$, i.e., $\ket{\psi} = \sum_{\alpha' } c_{\alpha'} \ket{v_{\alpha'}}$ for some coefficients $\{c_{\alpha'}\}_{\alpha'}$. Now, 
\begin{eqnarray}
    Q_{\alpha} \hat{O} \ket{\psi} &=& \sum_{\alpha'} c_{\alpha'}Q_{\alpha} \hat{O}  \ket{v_{\alpha'}} \nonumber\\ 
    &=& \sum_{\alpha'} c_{\alpha'} Q_{\alpha} \ket{w_{\alpha'}}   \nonumber\\ 
    &=& c_{\alpha} \ket{w_{\alpha}} \nonumber
\end{eqnarray} 
where above we have used the fact $\ket{w_{\alpha'}}:= \hat{O} \ket{v_{\alpha'}}$ is also an element in $\M{H}_{\alpha'}$ as $\hat{O}$ leaves the space $\M{H}_{\alpha'}$ invariant. On the other hand
\begin{eqnarray}
     \hat{O}Q_{\alpha} \ket{\psi} &=& \sum_{\alpha'} c_{\alpha'}\hat{O} Q_{\alpha} \ket{v_{\alpha'}}  \nonumber \\ 
    &=& c_{\alpha} \hat{O} \ket{v_{\alpha}}  \nonumber \\ 
    &=& c_{\alpha} \ket{w_{\alpha}}, \nonumber
\end{eqnarray}
as required.
\end{proof}
We now prove another technical result:
\begin{lemma}\label{lem: uniqueness of alpha0} Consider the Jordan decomposition of $\M{H}_{B} = \bigoplus_{\alpha}\M{H}_{\alpha}$ defined in terms of operators $\{P_0,\Pi_0\}$. 
  All spaces $\M{H}_{\alpha}$ are contained within the null space of $\Pi_0$, except for a single subspace $\M{H}_{\alpha^*}$. Furthermore, the projector onto $\M{H}_{\alpha^*}$ takes the form
\begin{eqnarray}
Q_{\alpha^*} = \Pi_0 + \bar{Q},
\end{eqnarray}
where $\bar{Q}$ a projector of rank at most 1, satisfying $\bar{Q}\Pi_0=0$. 
\end{lemma}
\begin{proof}
Taking $Q_\alpha$ to be the projector onto the subspace $\M{H}_\alpha$, from Lemma~\ref{lem: linear-algebra} we can deduce that $[\Pi_0, Q_{\alpha}] = 0$ for all $\alpha$, indicating that $Q_{\alpha}$ and $\Pi_0$ share common eigenvectors. Likewise, $[\id - \Pi_0, Q_{\alpha}] = 0$ implies that $\id - \Pi_0$ shares eigenvectors with $Q_{\alpha}$. Since $Q_{\alpha}$ is a projection onto a subspace of dimension at most 2, it must take the form $Q_{\alpha} = \gamma_{\alpha} \Pi_0 + \beta_{\alpha} \bar{P}_{\alpha}$, where $\bar{P}_{\alpha}$ is a projection onto a subspace of the null space of $\Pi_0$, and $\gamma_{\alpha}, \beta_{\alpha} \in \mathbb{R}$.

Any two distinct sub-spaces, $\M{H}_{\alpha}$ and $\M{H}_{\alpha'}$, must have orthogonal supports so $\gamma_{\alpha} \gamma_{\alpha'} = \gamma_{\alpha} \delta_{\alpha,\alpha'}$. Thus, $\gamma_{\alpha}$ can only be non-zero for one value of $\alpha$, which we call $\alpha^*$. Furthermore, $\gamma_{\alpha^*}= 1$ as $\gamma_{\alpha^*}^2 = \gamma_{\alpha^*}$ must hold. Now, as $Q_{\alpha^*}^2 = Q_{\alpha^*}$, we get $\beta_{\alpha^*}^2 = \beta_{\alpha^*}$, implying that $\beta_{\alpha^*} \in \{ 0 ,1\}$. Writing $\bar{Q} = \beta_{\alpha^*} \bar{P}_{\alpha^*}$ proves the lemma. 
\end{proof} 
The advantage of employing Jordan's lemma is that instead of considering a Hilbert space of unknown dimension, it allows us to decompose the space into orthogonal two-dimensional subspaces. As we shall show below, these two-dimensional subspaces can then be treated independently, thus allowing us to relax the optimization problem for $G(\score , \overlap)$ into an optimization problem on a two-dimensional Hilbert space. To illustrate this, we compute the overlap, score and the entropy for any $\rho \in \Gamma_{s}[\score , \overlap]$ (cf.~\eqref{eqn: single CQ state}).

\subsubsection{Computing the overlap}
To compute the score and the overlap of a state $\rho_{TXYE}$ of the form~\eqref{eqn: single CQ state}, we use the block structure defined by the projectors $Q_\alpha$ on to the subspaces $\M{H}_\alpha$ from the Jordan decomposition in terms of $\{M^0_0=P_0,M^1_0=\Pi_0\}$. We have
\begin{eqnarray}
 p_{Y|x,t}(y) := \tr \left( (M^{t}_{y}\otimes\id_E) \sigma^{x}_{BE} \right) &=& \tr \left( (\sum_{\alpha'} Q_{\alpha'}^2)M^{t}_{y}   (\sum_{\alpha} Q_{\alpha}^2) \sigma^{x}_{B} \right)  \label{eqn: p1}\\ 
              &=& \sum_{\alpha,\alpha'}\tr \left( Q_{\alpha'}M^{t}_{y} Q_{\alpha'} Q^2_{\alpha}  \sigma^{x}_{B}\right) \label{eqn: p2} \\ 
              &=& \sum_{\alpha} \tr \left( (Q_{\alpha} M^{t}_{y} Q_{\alpha})  (Q_{\alpha} \sigma^{x}_{B} Q_{\alpha})\right), \label{eqn: p3}
\end{eqnarray}
where we have used Lemma \ref{lem: linear-algebra} to commute $Q_{\alpha'}$ with $M^t_y$, the relation $Q_\alpha Q_{\alpha'}=Q_\alpha\delta_{\alpha,\alpha'}$ and the cyclic properties of the trace. The overlap is based on $p_{Y|x,0}(0)$. Noting that $M^{1}_{0} = \Pi_{0}$ gives
\begin{eqnarray}
p_{Y|x,0}(0) =  \sum_{\alpha} \tr \left( (Q_{\alpha} \Pi_0 Q_{\alpha}) (Q_{\alpha} \sigma^{x}_{B} Q_{\alpha})\right) =  \tr \left( (Q_{\alpha^*}\Pi_0 Q_{\alpha^*}) (Q_{\alpha^*} \sigma^{x}_{B} Q_{\alpha^*})\right)
\end{eqnarray}
and hence 
\begin{eqnarray}
\overlapfunction(\rho_{TXYE}) = \sum_{x} \frac{1}{2} \tr \left( (Q_{\alpha^*}\Pi_0 Q_{\alpha^*}) (Q_{\alpha^*} \sigma^{x}_{B} Q_{\alpha^*})\right) =  \sum_{x} \frac{1}{2} \tr \left( \Pi_0 (Q_{\alpha^*} \sigma^{x}_{B} Q_{\alpha^*})\right), \label{eqn: overlap} 
\end{eqnarray}
so that the overlap of $\rho_{TXYE}$ comes solely from the subspace $\M{H}_{\alpha^*}$. This motivates us to define
\begin{eqnarray}
\eta_{x} &:=& \tr(Q_{\alpha^*} \sigma^{x}_{B} Q_{\alpha^*} ),  \label{eqn: re-par-prob} \\ 
\tilde{\sigma}^{x}_{BE} &:=& \frac{1}{\eta_{x}} \left((Q_{\alpha^*} \ot \id_{E}) \sigma^{x}_{BE} (Q_{\alpha^*} \ot \id_{E})  \right),\quad\text{and} \label{eqn: re-par state} \\ 
\tilde{M}_{y} &:=& Q_{\alpha^*} M_{y}^{0} Q_{\alpha^*}. \label{eqn: re-par measurement}
\end{eqnarray}
Here $\eta_{x} \in [0 ,1]$, $\tilde{\sigma}^{x}_{B} \in \M{S}(\mathbb{C}^{2}) $ is a normalized qubit state and $\{\tilde{M}_{y} \}_{y}$ is a projective measurement on $\mathbb{C}^{2}$. Furthermore, using Lemma \ref{lemm: reduction to pure states} the state $\tilde{\sigma}^{x}_{BE}$ can be taken to be pure, and hence to be a state on $\M{S}(\mathbb{C}^{2} \ot \mathbb{C}^{2})$. The overlap Eq.~\eqref{eqn: overlap} can hence be written in terms of qubit states
\begin{eqnarray}\label{eqn: qubit overlap}
    \overlapfunction(\rho_{TXYE}) = \sum_{x} \frac{1}{2} \eta_{x} \tr(\Pi_{0} \tilde{\sigma}^{x}_{B}).
\end{eqnarray}

\subsubsection{Computing the score}
We can compute the score using the same technique as above
\begin{eqnarray}
  \scorefunction(\rho_{TXYE}) &=&  \sum_{x,\alpha} \frac{1}{2} \tr \left( (Q_{\alpha}M^{0}_{x} Q_{\alpha}) (Q_{\alpha} \sigma^{x}_{B} Q_{\alpha})\right)\\
  &=& \sum_{x} \frac{1}{2} \tr \left( (Q_{\alpha^*} {M^0_x} Q_{\alpha^*}) (Q_{\alpha^*} \sigma^{x}_{B} Q_{\alpha^*})\right) + \sum_{x,\alpha \ne \alpha^*}\frac{1}{2} \tr \left( (Q_{\alpha}M^{0}_{x} Q_{\alpha}) (Q_{\alpha} \sigma^{x}_{B} Q_{\alpha})\right). \nonumber 
\end{eqnarray}
Since $ (Q_{\alpha^*} \sigma^{x}_{B} Q_{\alpha^*}) \succeq 0$ and $(Q_{\alpha^*} M^{0}_{x} Q_{\alpha^*}) \succeq 0$, it follows that $\tr \left( (Q_{\alpha^*}M^{0}_{x} Q_{\alpha^*}) (Q_{\alpha^*} \sigma^{x}_{B} Q_{\alpha^*})\right) \geq 0$.  Therefore, it follows that $\scorefunction(\rho_{TXYE}) \geq \sum_{x} \frac{1}{2} \tr \left( (Q_{\alpha^*}{M^0_x} Q_{\alpha^*}) (Q_{\alpha^*} \sigma^{x}_{B} Q_{\alpha^*})\right)$. Furthermore, 
\begin{eqnarray}
 \sum_{\alpha \ne \alpha^*} \tr \left( (Q_{\alpha }M^{0}_{x} Q_{\alpha }) (Q_{\alpha } \sigma^{x}_{B} Q_{\alpha^*})\right) &\leq& \sum_{\alpha \ne \alpha^*} \tr \left( Q_{\alpha} \sigma^{x}_{B} Q_{\alpha} \right)  \nonumber \\ 
 &=& \sum_{\alpha } \tr \left( Q_{\alpha} \sigma^{x}_{B} Q_{\alpha}\right) -  \tr \left( Q_{\alpha^*} \sigma^{x}_{B} Q_{\alpha^*}\right) \nonumber \\ 
&=&  1 -   \tr \left( Q_{\alpha^*} \sigma^{x}_{B} Q_{\alpha^*} \right).
 \nonumber
\end{eqnarray}
Therefore, we can also bound $\scorefunction(\rho_{TXYE})$ solely in terms of a single qubit strategy using
\begin{eqnarray}
 \frac{1}{2}\sum_{x}  \tr \left( (Q_{\alpha^*}M^{0 }_{x} Q_{\alpha^*}) (Q_{\alpha^*} \sigma^{x}_{B} Q_{\alpha^*}) \right) \leq  \scorefunction(\rho_{TXYE}) \leq \frac{1}{2}\sum_{x}\left(\tr \left((Q_{\alpha^*}M^{0 }_{x} Q_{\alpha^*}) (Q_{\alpha^*} \sigma^{x}_{B} Q_{\alpha^*})\right)  +  1 -\tr \left( Q_{\alpha^*} \sigma^{x}_{B} Q_{\alpha^*} \right) \right).  \nonumber
\end{eqnarray}
Rewriting the equation above in terms of the parametrization \eqref{eqn: re-par-prob}--\eqref{eqn: re-par measurement} gives the bound:
\begin{eqnarray}\label{eqn: qubit-score}
\sum_{x} \frac{\eta_{x}}{2}  \tr \left( \tilde{M}_{x} \tilde{\sigma}^{x}_{B}   \right) \leq \scorefunction(\rho_{TXYE}) \leq \sum_{x} \left( \frac{\eta_{x}}{2}  \tr \left( \tilde{M}_{x} \tilde{\sigma}^{x}_{B} \right) + \frac{1 - \eta_{x}}{2} \right).
\end{eqnarray}

\subsubsection{Lower bounding the entropy} 
Using arguments similar to those used in \eqref{eqn: p1}--\eqref{eqn: p3} we get the following decomposition of $\rho_{TXYE}$:
\begin{eqnarray}
\rho_{TXYE} &=& \sum_{t,x,y} p_{T}(t) p_{X}(x)\ketbra{t,x,y}{t,x,y}  \otimes \underset{B}{\tr} \left( (M_{y}^{t} \otimes \id_{E}) \sigma^{x}_{BE}    \right) \nonumber\\ 
            &=& \sum_{t,x,y} p_{T}(t) p_{X}(x)\ketbra{t,x,y}{t,x,y}  \otimes \underset{B}{\tr}\left(\sum_{\alpha} ((Q_{\alpha} M_{y}^{t} Q_{\alpha}) \otimes \id_{E}) (Q_{\alpha} \ot \id_{E}) \sigma^{x}_{BE} (Q_{\alpha} \ot \id_{E} ).    \right) \nonumber \\ 
            &=& \sum_{x ,\alpha} \tr\left( Q_{\alpha}  \sigma^{x}_{B} Q_{\alpha} \right)p_{X}(x)\left( \sum_{t,y}  p_{T}(t) \ketbra{t,x,y}{t,x,y}  \otimes \underset{B}{\tr}\left( ((Q_{\alpha} M_{y}^{t} Q_{\alpha}) \otimes \id_{E}) \left(\frac{(Q_{\alpha} \ot \id_{E}) \sigma^{x}_{BE} (Q_{\alpha} \ot \id_{E})}{\tr\left( Q_{\alpha}  \sigma^{x}_{B} Q_{\alpha} \right)} \right)    \right) \right)   \nonumber 
\end{eqnarray}  
Defining
$$\rho^{\alpha}_{T,X=x,YE}  :=  \sum_{t,y} p_{T}(t)  \ketbra{t,x,y}{t,x,y}  \otimes \underset{B}{\tr}\left( ((Q_{\alpha} M_{y}^{t} Q_{\alpha}) \otimes \id_{E}) \left(\frac{(Q_{\alpha} \ot \id_{E}) \sigma^{x}_{BE} (Q_{\alpha} \ot \id_{E})}{\tr\left( Q_{\alpha}  \sigma^{x}_{B} Q_{\alpha} \right)} \right)    \right), $$ 
we can express $\rho_{TXYE} = \sum_{x, \alpha} \tr(Q_{\alpha} \rho^{x}_{B} Q_{\alpha}) p_{X}(x) \rho^{\alpha}_{T,X=x,YE}$. We now simplify the expression for the entropy:\footnote{Note that it can be easily verified that
$H(Y|X = x, T = 0, E)_{\rho_{TXYE}} = H(Y|X = x, Y = 0, E)_{\rho_{T,X=x,Y,E}}.
$}
\begin{eqnarray}
H(Y|X, T= 0, E)_{\rho_{TXYE}} &=&  \sum_{x} p_{X}(x) H(Y|X=x, T=0 ,E)_{\rho_{TXYE}} \nonumber \\ 
               &\geq& \sum_{x} p_{X}(x) \tr(Q_{\alpha}^* \rho^{x}_{B} Q_{\alpha}^*)  H(Y|X=x,T=0,E)_{\rho^{\alpha^*}_{TXYE}} \nonumber \\ 
               &\geq&  \sum_{x} p_{X}(x) \eta_{x} H(Y|X=x,T=0,E)_{\rho^{\alpha^*}_{TXYE}}. \label{eqn: qubit_CQ_entropy}
\end{eqnarray}

\subsubsection{Entropy for a qubit strategy}
We now explicitly compute  the entropy $H(Y|X=x,T=0,E)_{\rho^{\alpha^*}_{TXYE}}$. Using the chain rule for conditional von Neumann entropy
\begin{eqnarray*}
H(Y|X=x , T= 0 , E) &=& H(YE|X=x , T=0)  - H(E|X=x , T=0)\\
&=& H(Y|X=x,T=0) + H(E|Y, X=x , T= 0) - H(E | X=x, T= 0).
\end{eqnarray*}
We then compute
\begin{eqnarray}
H(Y|X=x,T=0)_{\rho^{\alpha^*}_{TXYE}}   = H_{\bin}\left( \tr \left(\tilde{\sigma}^{x}_{B} \tilde{M}_{0} \right)\right).
\end{eqnarray}
Similarly, a straightforward calculation for $H(E|X=x,T=0)$ shows that $H(E|X=x, T=0) = H(\tilde{\sigma}^{x}_{E})$. Finally, because the state $\tilde{\sigma}^{x}_{BE}$ can be taken to be pure, we have $H(\tilde{\sigma}^{x}_{E}) = H(\tilde{\sigma}^{x}_{B})$.

Finally, we note that for any CQ state $H(E|Y,X=x , T=0)$ is non-negative, giving the bound
\begin{eqnarray}\label{eqn: qubit_entropy}
 H(Y|X=x,T=0,E) \geq H_{\bin}\left( \tr (\tilde{\sigma}^{x}_{B} \tilde{M}_{0})\right) - H(\tilde{\sigma}^{x}_{B}).
\end{eqnarray}
Since the term $H(E|Y,X=x,T=0)$ is zero when $\tilde{\sigma}^{x}_{BE}$ is a pure state and $\{ \tilde{M}_{y} \}_{y}$ is a rank one projective measurement (see Appendix~C of~\cite{BhRC}) this bound is achievable.

Summarizing the above results gives us the following lower bound on $G(\score , \overlap)$.
\begin{lemma}\label{lemm: Jodran's decomposition corollary}
For every $(\score , \overlap) \in [1/2 , 1] \times [1/2 ,1]$, $G_{p_{X}}(\score , \overlap) \geq \tilde{G}_{p_{X}}(\score, \overlap)$, where $\tilde{G}_{p_{X}}(\score, \overlap)$ is the solution to the optimization problem 
\begin{equation}\label{eqn: Gopt}
  \begin{aligned} 
     \inf \quad &  \sum_{x} \eta_{x} p_{X}(x) \left(H_{\bin}\left( \tr \left( P_{0} \tilde{\sigma}^{x}_{B}  \right)\right) - H(\tilde{\sigma}^{x}_{B})\right) \\
\mathrm{s.t.} 
\quad &  \forall x \in \{ 0 , 1 \}: 0 \leq \eta_{x} \leq 1\\
\quad & \forall x: \tilde{\sigma}^{x}_{B} \in \M{S}(\mathbb{C}^{2}) \\
\quad & \forall x: P_{x} = \proj{\psi_x}, \ket{\psi_x} \in \mathbb{C}^{2} , \braket{\psi_x}{\psi_x} =1 \\ 
\quad & \sum_{x} P_{x} = \id \\
\quad &  \sum_{x} \frac{\eta_{x}}{2}  \tr \left( P_{x} \tilde{\sigma}^{x}_{B}   \right) \leq \score \leq \sum_{x} \left( \frac{\eta_{x}}{2}  \tr \left( P_{x} \tilde{\sigma}^{x}_{B} \right) + \frac{1 - \eta_{x}}{2} \right) \\
              \quad &        \sum_{x} \frac{1}{2} \eta_{x} \tr(\Pi_{0} \tilde{\sigma}^{x}_{B}) = \overlap.
 \end{aligned}  
\end{equation}
\end{lemma}
\begin{proof}
The proof follows immediately from Equations~\eqref{eqn: qubit overlap},~\eqref{eqn: qubit-score},~\eqref{eqn: qubit_CQ_entropy} and~\eqref{eqn: qubit_entropy}. 
\end{proof}

\subsection{Simplifying the qubit optimization problem} 
In this section, we simplify the optimization problem \eqref{eqn: Gopt}. Our primary objective is to reformulate the problem into an optimization over a reduced number of bounded real variables. This transformation not only reduces the complexity of the problem but also renders it tractable for solving using known numerical techniques. 
\subsubsection{Optimization problem in terms of bounded real variables}
The optimization in Lemma \ref{lemm: Jodran's decomposition corollary} is expressed in terms of the states $\tilde{\sigma}^{x}_{B}$ and the projection operator $P_0$. We would like to re-express it in terms of a standard optimization problem on a bounded domain $\mathcal{D} \subseteq \mathbb{R}^{n}$. To achieve this, we write
\begin{align}\label{eqn: parametrization_of_states_and_measurements}
\tilde{\sigma}_{B}^{x}  = \frac{\id}{2} + \sum_{i=1}^{3} \frac{a_i^{x}}{2} \hat{\boldsymbol{\sigma}}_i , \quad  P_0  = \ketbra{\psi_0}{\psi_0} = \frac{\id}{2} + \sum_{i=1}^{3} \frac{b_i}{2} \hat{\boldsymbol{\sigma}}_i, \text{\ \ and\ \ } \Pi_0 = \ketbra{0}{0} = \frac{\id}{2} + \frac{\hat{\boldsymbol{\sigma}}_3}{2}, 
\end{align}
where $\hat{\boldsymbol{\sigma}_{1}} , \hat{\boldsymbol{\sigma}_2}$ and $\hat{\boldsymbol{\sigma}_3}$ are Pauli $x$, $y$ and $z$ respectively and $\sum_{i} (a_i^x)^2 \leq 1$ and $\sum_{i} b_i^2 = 1$ must be satisfied to ensure that the parameters $\{a_i^x\}_{i=1}^{3}$ and $\{b_i\}_{i=1}^{3}$ represent a valid density operator and a projective measurement, respectively. We now have an optimization problem over $11$ variables: 3 parameters for each state $\tilde{\sigma}^{x}_{B}$, three parameters describing the measurement operator and two parameters $\eta_0$ and $\eta_1$ that give the probability of the Jordan block $\M{H}_{0}$ occurring. Thus, the optimization problem can be cast as an optimization problem over a compact subset of $\mathbb{R}^{d}$, where $d=11$ at this stage. We now reduce the value of $d$ to simplify the optimization problem by removing some redundant variables.

Before simplifying this optimization problem further, we introduce the function $\Phi:[-1,1]\to\mathbb{R}$ given by
\begin{eqnarray}
  \Phi(x):=H_{\bin}\left(\frac{1}{2}+\frac{x}{2}\right)
\end{eqnarray}
to keep the expressions more compact.
Several properties of $\Phi(x)$ such as its monotonicity for $x > 0$ can be directly inferred from the properties of the binary entropy $H_{\bin}$. Some other useful properties of $\Phi(x)$ that are relevant to our proof are discussed in Appendix~\ref{app: useful claims}. We now give our first simplification.
\begin{lemma}
The optimization problem \eqref{eqn: Gopt} has the same value as the following optimization problem
\begin{eqnarray}\label{eqn: Gopt_standard}
 \begin{aligned}
     \min \quad &  \sum_{x}   \eta_{x} p_{X}(x) \left(\Phi\left( a_{x} \cos(\theta_x - \phi)\right) - \Phi \left(a_{x} \right)\right)  \\
\mathrm{s.t.}
\quad & \forall x \in \{ 0 , 1\} : \eta_{x} , a_{x} \in [0 , 1] \\
\quad & \forall x \in \{ 0 , 1 \}: \theta_{x} \in [0 , 2 \pi] \\
\quad & \phi \in [0 , 2 \pi] \\
\quad &  \sum_{x}  \frac{1}{2} \eta_{x} \left( \frac{1}{2} + \frac{(-1)^{x} a_{x} \cos(\theta_x - \phi)}{2}\right) \in [\score - \sum_{x} \frac{1}{2}(1 - \eta_{x}), \score]   \\
              \quad &   \sum_{x} \frac{1}{2}\eta_{x}  \left( \frac{1}{2} + \frac{a_{x} \cos(\theta_{x})}{2}\right) = \overlap .
 \end{aligned} 
\end{eqnarray}
\end{lemma}
\begin{proof}
First, note that the objective function and constraints in the optimization problem \eqref{eqn: Gopt} can be expressed solely in terms of $\tr\left(P_{0} \tilde{\sigma}^{x}_{B} \right)$, $\tr\left( \Pi_{0} \tilde{\sigma}^{x}_{B} \right)$, and $H(\tilde{\sigma}^{x}_{B})$. Using the parametrization \eqref{eqn: parametrization_of_states_and_measurements} we can compute the objective function and the constraints:
\begin{eqnarray}
\tr \left(P_{0} \tilde{\sigma}^{x}_{B}\right) &=& \frac{1}{2} + \frac{\B{a}^{x}.\B{b}}{2} \nonumber \\
\tr \left(\Pi_0 \tilde{\sigma}^{x}_{B}\right) &=& \frac{1}{2} + \frac{\B{a}^{x}.\hat{\B{z}}}{2} .
 \nonumber 
\end{eqnarray}
Without loss of generality, we can choose a basis such that $\hat{\B{z}} = (0 , 0 , 1)$ and $\B{b} = (\sin(\phi) , 0 , \cos(\phi))$. Furthermore, we write $\B{a}^{x} = ( a_x \sin(\theta_x) , a_x' ,  a_x \cos(\theta_x))$, taking $a_x\in[0,1]$ to give 
\begin{eqnarray}
\tr \left(P_{0} \tilde{\sigma}^{x}_{B}\right) &=& \frac{1}{2} + \frac{a_{x} \cos(\theta_{x} - \phi)}{2} \nonumber \\
\tr \left(\Pi_0 \tilde{\sigma}
^{x}_{B}\right) &=& \frac{1}{2} + \frac{a_{x} \cos(\theta_{x})}{2}. 
 \nonumber 
\end{eqnarray}
The entropy $H(\tilde{\sigma}^{x}_{B})$ can be computed in terms of its eigenvalues $\frac{1}{2} + \frac{|\B{a}^{x}|}{2}$ and $\frac{1}{2} - \frac{|\B{a}^{x}|}{2}$ as
\begin{eqnarray}
H(\tilde{\sigma}_{B}^{x}) = H_{\bin}\left(\frac{1}{2} + \frac{|\B{a}^{x}|}{2} \right) = \Phi\left(|\B{a}^{x}| \right). \nonumber
\end{eqnarray}
We can use the monotonicity of $\Phi(x)$ for $x>0$ to give 
\begin{eqnarray}
   H(\tilde{\sigma}_{B}^{x})  = \Phi\left(\sqrt{a_x^2 + a_x'^2 }\right)\leq\Phi(a_x), \nonumber
\end{eqnarray}
with equality if $a'_x=0$. Since $\tr\left(P_0\sigma^x_B\right)$ and $\tr\left(\Pi_0 \sigma^x_B\right)$ are independent of $a_x'$, setting $a'_x$ to zero has no effect on the constraints, and hence the optimum occurs when $a'_x=0$. 

Finally, note that $\eta_x \in [0,1]$, $a_x \in [0,1]$, and $\theta_x \in [0,2\pi]$, $\phi \in [0,2\pi]$ can all be taken to lie in closed and bounded sets. Consequently, the feasible region of the optimization problem is compact. Since the objective function is continuous over this region, the optimum is always attained. In particular, the minimum exists, and therefore the infimum in the optimization problem may be replaced by a minimum.

Combining the above results and performing the substitutions in \eqref{eqn: Gopt}, establishes the claim.
\end{proof}

\subsubsection{Extending the domain} 
To compute the rates, we want to solve the optimization problem \eqref{eqn: Gopt_standard} using numerical techniques. We find that this optimization problem is too resource intensive to solve directly, and therefore, we aim to lower bound this problem, which can be effectively handled numerically. By Lemma~\ref{Claim: monotonicity in x} (see Appendix~\ref{app: useful claims}), we can lower bound the objective function by noting that 
\begin{eqnarray}
    \eta_{x} \left(\Phi\left( a_{x} \cos(\xi_{x})\right) - \Phi \left(a_{x} \right)\right)  \geq \Phi\left( \eta_{x} a_{x} \cos(\xi_{x})\right) - \Phi \left(\eta_{x} a_{x} \right) .
\end{eqnarray} 
Furthermore, we can replace the constraints 
\begin{eqnarray}
  \sum_{x}  \frac{1}{2} \eta_{x} \left( \frac{1}{2} + \frac{(-1)^{x} a_{x} \cos(\theta_x - \phi)}{2}\right) &\in& [\score - \sum_{x} \frac{1}{2}(1 - \eta_{x}), \score]  \nonumber  \\
  \sum_{x} \frac{1}{2}\eta_{x}  \left( \frac{1}{2} + \frac{a_{x} \cos(\theta_{x})}{2}\right) &=& \overlap \nonumber
\end{eqnarray}
of the optimization problem \eqref{eqn: Gopt_standard} by the relaxed constraints
\begin{eqnarray}
  \sum_{x}  \frac{1}{2} \eta_{x} \left( \frac{1}{2} + \frac{(-1)^{x} a_{x} \cos(\theta_x - \phi)}{2}\right) &\geq& \score - \frac{1}{2} \sum_{x}(1 - \eta_{x})   \nonumber  \\
  \sum_{x} \frac{1}{2}\eta_{x}  \left( \frac{1}{2} + \frac{a_{x} \cos(\theta_{x})}{2}\right) &\geq&  \overlap. \nonumber
\end{eqnarray}
This relaxation cannot increase the minimal value, as the feasible set of \eqref{eqn: Gopt_standard} is a subset of the relaxed feasible set. 

For convenience, we replace $a_x$ with $\tilde{a}_x=\eta_x a_x$ and the constraint $a_x \in [0, 1]$ with $\tilde{a}_x\leq\eta_x$, which eliminates $a_x$ from the optimization. In addition, we write $\xi_x=\theta_x - \phi$. The next Lemma summarizes the above.
\begin{lemma}\label{lemm: one more lower bound}
A lower bound on the function $\tilde{G}_{p_{X}}(\score, \overlap) $ can be achieved by computing the optimization problem 
\begin{eqnarray}\label{eqn : Non-polynomial optimization problem}
 \begin{aligned} \entropy_{p_{X}}(\score , \overlap) :=   
     \inf \quad &  \sum_{x}  p_{X}(x) \left(\Phi\left( \tilde{a}_{x} \cos(\xi_{x})\right) - \Phi \left(\tilde{a}_{x} \right)\right)  \\
\mathrm{s.t.} 
\quad &  \forall x \in \{ 0 , 1 \}: 1 \geq \eta_{x} \geq \tilde{a}_{x} \geq 0 \\
\quad & \forall x: \xi_{x} \in [0 , 2 \pi] \\
\quad&\phi\in[0,2\pi]\\
\quad &  \sum_{x}    \left( -\eta_{x} + (-1)^{x}  \tilde{a}_{x} \cos(\xi_{x})\right) \geq 4 \score -4  \\
              \quad &   \sum_{x}   \left( \eta_{x} + \tilde{a}_{x} \cos(\xi_{x} + \phi)\right) \geq  4 \overlap. 
 \end{aligned}  
\end{eqnarray} 
\end{lemma} 

We now note two important properties of \(\entropy_{p_{X}}(\score , \overlap)\).

\begin{lemma}
\label{lem:monotonicity-zero}
\(\entropy_{p_X}(\score,\overlap)\) is nondecreasing in \(\score,\overlap\). Moreover, \(\M{G}_{p_X}(\score,\overlap)=0\) whenever \(\score\le\tfrac{1}{2}\) or \(\overlap\le\tfrac{1}{2}\).
\end{lemma}

\begin{proof}
Note that the values of $\score$ and $\overlap$ do not appear in the objective function of optimization problem \eqref{eqn : Non-polynomial optimization problem}. Thus, it suffices to show that the set of feasible points for the optimization problem $\entropy_{p_{X}}(\score + \delta \score , \overlap + \delta \overlap)$ is a subset of the feasible set for the optimization problem for $\entropy_{p_{X}}(\score  + \delta \score , \overlap)$. This follows because
\begin{eqnarray*}
    \sum_{x}    \left( -\eta_{x} + (-1)^{x} \tilde{a}_{x} \cos(\xi_{x})\right)) \geq  4 (\score + \delta \score) -4 \geq 4(\score) - 4   \\ 
    \sum_{x}   \left( \eta_{x} + \tilde{a}_{x} \cos(\xi_{x} + \phi)\right) \geq  4 (\overlap + \delta \overlap) \geq 4 \delta \overlap,
\end{eqnarray*}
and hence \(\entropy_{p_X}(\score,\overlap)\) is nondecreasing in \(\score,\overlap\).

For the second statement, note that if there exist feasible parameters with \(\xi_0,\xi_1\in\{0,\pi\}\), then the objective function equals zero, i.e., \(\entropy_{p_X}(\score,\overlap)=0\). \\ 

We first show \(\M{G}_{p_X}(\score,\tfrac{1}{2})=0\) for all \(\score\in[0,1]\).
If \(\score\ge\tfrac{1}{2}\), take \(\eta_x=1\), \(\tilde a_0=\tilde a_1=a=2\score-1\), \(\xi_0=0\), \(\xi_1=\pi\), and any \(\phi\).
Then
\(
\sum_x\bigl(\eta_x+\tilde a_x\cos(\xi_x+\phi)\bigr)=2=4\cdot\tfrac{1}{2},
\)
and
\(
\sum_x\bigl(-\eta_x+(-1)^x\tilde a_x\cos\xi_x\bigr)
= 2(a-1)=4\score-4,
\)
so feasibility holds with \(\overlap=\tfrac{1}{2}\) and \(\score \geq \frac{1}{2}\). \\ 
If \(\score\le\tfrac{1}{2}\), take \(\eta_x=1\), \(\tilde a_0=\tilde a_1=a=1-2\score\), \(\xi_0=\pi\), \(\xi_1=0\), and \(\phi=\pi\).
Then again
\( 
\sum_x\bigl(\eta_x+\tilde a_x\cos(\xi_x+\phi)\bigr)=2=4\cdot\tfrac{1}{2},
\)
and
\(\sum_x\bigl(-\eta_x+(-1)^x\tilde a_x\cos\xi_x\bigr)
= -2(1+a)=4\score-4,
\) 
so \(\M{G}_{p_X}(\score,\tfrac{1}{2})=0\) for all \(\score\in[0,1]\). \\

Similarly, \(\M{G}_{p_X}(\tfrac{1}{2},\overlap)=0\) for all \(\overlap\in[0,1]\).
Take \(\eta_x=1\), \(\xi_0=\xi_1=0\).
Then the score constraint gives
\[
\sum_x\bigl(-\eta_x+(-1)^x\tilde a_x\cos\xi_x\bigr)= 4\cdot \frac{1}{2}-4,
\]
i.e., \(\score=\tfrac{1}{2}\).
For the overlap constraint:
\begin{itemize}
\item If \(\overlap\ge\tfrac{1}{2}\), set \(\phi=0\), \(\tilde a_0=\tilde a_1=a=2\overlap-1\):
\(\sum_x(\eta_x+\tilde a_x\cos(\xi_x+\phi)) =2+2a=4\overlap\).
\item If \(\overlap\le\tfrac{1}{2}\), set \(\phi=\pi\), \(\tilde a_0=\tilde a_1=a=1-2\overlap\):
\(\sum_x(\eta_x+\tilde a_x\cos(\xi_x+\phi))= 2-2a=4\overlap\).
\end{itemize}
Thus \(\M{G}_{p_X}(\tfrac{1}{2},\overlap)=0\) for all \(\overlap\).

By monotonicity in both $\score$ and $\overlap$ we have that, \(\M{G}_{p_X}(\score,\overlap)=0\) whenever \(\score\le\tfrac{1}{2}\) or \(\overlap\le\tfrac{1}{2}\).
\end{proof}

\subsubsection{Eliminating an additional variable} 

In this subsection we eliminate $\phi$ from the optimization problem. To do so, we rely on the fact that $\phi$ only features in one constraint. The general result is given in Lemma~\ref{lemm: maximizing the constraint} before being applied to our case. 
\begin{lemma}\label{lemm: maximizing the constraint}
Let $I$ be a compact subset of $\mathbb{R}$. Consider the optimization problem  
\begin{eqnarray}\label{eqn: general optimization : 1}
 \begin{aligned}
     \min \quad &  f(x_1 ,x_2 , \cdots , x_{n}) \\
\mathrm{s.t.}
\quad & \forall i: h_i(x_1 , x_2 , \cdots , x_{n}) \geq 0\\
\quad &  g(x_1 , x_2 , \cdots , x_{n} , y) \geq 0\\
\quad & y\in  I
\end{aligned}
\end{eqnarray}
is equivalent to
\begin{eqnarray}\label{eqn: general optimization}
 \begin{aligned}
     \min \quad &  f(x_1 ,x_2 , \cdots , x_{n}) \\
\mathrm{s.t.}
\quad & \forall i: h_i(x_1 , x_2 , \cdots , x_{n}) \geq 0\\
\quad &  g(x_1 , x_2 , \cdots , x_{n} , y) \geq 0 \\
\quad &  y \in Y^\ast,
\end{aligned} 
\end{eqnarray}
where $Y^\ast$ is the set
\[
Y^\ast(x_1,\dots,x_n) \;:=\; \operatorname*{arg\,max}_{y \in I} \; g(x_1,\dots,x_n,y).
\]
\end{lemma}
\begin{proof}
Fix $(x_1,\dots,x_n)$ satisfying $h_i(x_1,\dots,x_n)\ge 0$ for all $i$.  
The constraint $g(x_1,\dots,x_n,y)\ge 0$ is feasible for some $y\in I$ if and only if $
\max_{y\in I} g(x_1,\dots,x_n,y) \;\ge\; 0$. 
Hence feasibility depends only on the maximal value of $g$ in the interval $I$.   
If there exists any $y\in I$ such that $g(x_1,\dots,x_n,y)\ge 0$, then there exists
$y^\ast \in Y^\ast(x_1,\dots,x_n)$ satisfying the constraint.

Since the variable $y$ does not appear in the objective function, restricting $y$ to
$Y^\ast(x_1,\dots,x_n)$ does not change the feasible set in the $(x_1,\dots,x_n)$ variables,
nor does it affect the optimal value of the problem. This proves the equivalence of
\eqref{eqn: general optimization : 1} and \eqref{eqn: general optimization}.
\end{proof}

We can apply Lemma~\ref{lemm: maximizing the constraint} to our problem. We can eliminate the dependence of $\phi$ by maximizing the constraint it appears in. Using
\begin{eqnarray*}
  \sum_{x} \tilde{a}_x \cos(\xi_{x} + \phi) &=&   \left( \sum_{x}\tilde{a}_x \cos(\xi_{x}) \right) \cos(\phi) - \left( \sum_{x} \tilde{a}_x \sin(\xi_{x}) \right)  \sin(\phi)
  \end{eqnarray*}
the maximum of $\sum_x \tilde{a}_x \cos(\xi_{x} + \phi)$ over $\phi$ is  
\begin{eqnarray*}
\sqrt{ \left( \sum_{x}\tilde{a}_x \cos(\xi_{x}) \right)^2 +  \left( \sum_{x} \tilde{a}_x \sin(\xi_{x}) \right)^2} .
\end{eqnarray*}
Note that $\overlap - \sum_{x} \eta_{x}\geq 0$ for $\overlap \geq \frac{1}{2}$ and  
\begin{eqnarray}
    \sqrt{ \left( \sum_{x}\eta_{x} a_{x} \cos(\xi_{x}) \right)^2 +  \left( \sum_{x} \eta_{x} a_{x} \sin(\xi_{x}) \right)^2}  + \sum_{x} \eta_{x} \geq 4 \overlap
\end{eqnarray}
imply
\begin{eqnarray}
    \left( \sum_{x}\eta_{x} a_{x} \cos(\xi_{x}) \right)^2 +  \left( \sum_{x} \eta_{x} a_{x} \sin(\xi_{x}) \right)^2  \geq \left( 4 \overlap - \sum_{x} \eta_{x}\right)^2 .
\end{eqnarray}
Summarizing the above leads to the following lemma.
\begin{lemma} The optimization problem \eqref{eqn : Non-polynomial optimization problem} has the same solution as
\begin{eqnarray}
\begin{aligned}\label{eqn: optimization_final}
\entropy_{p_{X}}(\score, \overlap) =      \inf \quad &  \sum_{x}    p_{X}(x) \left(\Phi\left(\tilde{a}_{x} \cos(\xi_{x})\right) - \Phi \left(\tilde{a}_{x} \right)\right)  \\
\mathrm{s.t.} \quad &  \sum_{x}    \left( -\eta_{x} + (-1)^{x} \tilde{a}_{x} \cos(\xi_{x})\right) \geq  4 \score -4  \\
              \quad &    \left(  \sum_{x}\tilde{a}_{x} \cos(\xi_{x}) \right)^2 +  \left( \sum_{x} \tilde{a}_{x} \sin(\xi_{x}) \right)^2 \geq \left( 4 \overlap  - \sum_{x} \eta_{x} \right)^2 \\ 
              \quad &  \forall x \in \{0 , 1 \}:  1 \geq \eta_{x} \geq \tilde{a}_{x} \geq 0,\\
              & \xi_x\in[0,2\pi].
 \end{aligned}
\end{eqnarray}
\end{lemma}

\subsection{Obtaining numerical bounds on the rates} 
The optimization problem \eqref{eqn: optimization_final} is non-linear and hence challenging to solve. In this subsection we introduce methods to compute lower bounds on it.

\subsubsection{Computing the optimization problem over grids}\label{app:grid}
We now address the problem of computing the convex envelope of the function $\entropy_{p_{X}}(\score, \overlap)$. As the functional form of this is difficult to obtain, we must compute this function numerically using known optimization techniques. However such computation can only happen over a finitely many points and not on the entire feasible set. 

We now address the problem of computing the convex envelope of the function \(\entropy_{p_X}(\score, \overlap)\). 
As the functional form of this function is difficult to obtain analytically, we must compute it numerically using known optimization techniques. 
However, such computations can only be performed on a finite set of points and not over the entire feasible set. 

We define the concept of a Gridding generally before specializing to our case. 
Let \(\mathcal{X}^{(i)}\) be a set of points 
\(\{x_1^{(i)}, x_2^{(i)}, \dots, x_{n_i}^{(i)}\} \subset [a^{(i)}, b^{(i)}]\) with $x_1^{(i)}=a^{(i)}$, $x_{n_i}^{(i)}=b^{(i)}$ and $x_j^{(i)}<x_{j+1}^{(i)}$ for all $i$ and $j\leq n_i-1$. 
We call 
\[
\mathcal{P} := \mathcal{X}^{(1)} \times \mathcal{X}^{(2)} \times \cdots \times \mathcal{X}^{(n)}
\] 
a \emph{gridding} of the (hyper)rectangle 
\([a^{(1)}, b^{(1)}] \times \cdots \times [a^{(n)}, b^{(n)}].\)

For a given gridding \(\mathcal{P}\) and a function 
\(F: [a^{(1)}, b^{(1)}] \times \cdots \times [a^{(n)}, b^{(n)}] \to \mathbb{R}\), 
define the \emph{grid-restricted function}
\begin{equation}
F^{\mathcal{P}}(\mathbf{x}) := F(\lfloor \B{x} \rfloor_{\M{P}}),
\end{equation}
where \(\lfloor \B{x} \rfloor_{\M{P}} \) is defined component-wise as
\[
\lfloor \mathbf{x} \rfloor_{\mathcal{P}} := \big( \max\{u \in \mathcal{X}^{(1)} \mid u \le x_1\}, \dots, \max\{u \in \mathcal{X}^{(n)} \mid u \le x_n\} \big),
\]

In other words, \(\lfloor\mathbf{x}\rfloor_{\M{P}}\) is the largest grid point in \(\mathcal{P}\) that does not exceed \(\mathbf{x}\) in any coordinate. 

For our case, we can consider a gridding of the set $[1/2 , 1 ] \times [1/2 ,1 ]$ and construct the function $\entropy_{p_X}^{\M{P}}(\score , \overlap)$ by numerically computing $\entropy_{p_{X}}(\score , \overlap)$ on the values of $\score$ and $\overlap$ in the gridding. Since $\entropy_{p_{X}}(\score , \overlap)$ is non-decreasing in $\score$ and $\overlap$, 
\(\entropy_{p_X}^{\M{P}}(\score , \overlap)\) is a valid lower bound on $\entropy_{p_{X}}(\score , \overlap)$ in the interval $[1/2 , 1] \times [1/2 , 1]$.

Using the algorithms described in Section~\ref{app: LF transform}, the function $\M{F}_{p_{X}}^{\M{P}}(\score , \overlap) := \text{convenv}(\entropy^{\M{P}}_{p_{X}}(\score , \overlap))$ can be computed over the gridding $\M{P}$ (in time linear in the number of points in the gridding), and the function can be extended to the entire domain using the interpolation technique mentioned above.

Note that \(\mathcal{F}_{p_X}^{\mathcal{P}}\) provides a lower bound on \(F_{p_X}\) (see Lemma~27 of~\cite{BhRC}). Furthermore, the optimization problem for \(\entropy_{p_X}^{\mathcal{P}}(\score, \overlap)\) can be relaxed by approximating it with a polynomial optimization problem. We discuss this relaxation in the next section.
\subsubsection{Converting the optimization problem to a polynomial optimization problem}\label{sec: poly-approximations of Phi(x)}
In this section, we approximate the function $\Phi(x)$ (and its difference) in terms of polynomials.
\begin{lemma}\label{lem: poly-approx}
The function $\Phi(x)$ has the following polynomial lower bounds: 
\begin{eqnarray}\label{eq:Phin}
 \Phi(x) \geq \Phi_{n}(x) := \sum_{k = 0}^{n} I_{k} \, x^{2k} (1 - x^2),
\end{eqnarray}
where 
\begin{eqnarray}\label{eq:Ik}
I_k:=\int_{\frac{1}{2}}^{1}\frac{1}{z\ln 2} \left(\frac{1-z}{z}\right)^{2k} \, \mathrm{d}z.
\end{eqnarray}
\end{lemma}

\begin{proof}
We begin with an integral representation of the logarithm~\cite{BFF}: 
\begin{eqnarray} 
\log(x) = \frac{1}{\ln 2} \int_{0}^{1} \frac{x - 1}{t(x - 1) + 1} \, \mathrm{d}t .
\end{eqnarray}

From this we obtain an integral representation of the binary entropy function 
$\Phi(x) := H_{\mathrm{bin}}\Big(\frac{1}{2} + \frac{x}{2}\Big)$:
\begin{eqnarray}
\Phi(x) = \frac{(1 - x^2)}{\ln 2}\int_{0}^{1} \frac{2 - t}{(2 - t(1 - x))(2 - t(1 + x))} \, \mathrm{d}t .
\end{eqnarray}
Changing variables to $z=1-t/2$, allows this to be rewritten as
\begin{eqnarray}
\Phi(x) = \int_{\frac{1}{2}}^{1} \frac{1}{z \ln 2} \, \frac{1 - x^2}{1 - \left( \frac{1 - z}{z} \right)^2 x^2} \, \mathrm{d}z.
\end{eqnarray}

Since $\left(\frac{x(1-z)}{z}\right)^2\in [0, 1]$ for $z \in [\frac{1}{2}, 1]$ and $x\in[-1, 1]$, the integrand can be expanded into the converging series 
\begin{eqnarray} 
\Phi(x) &=& \sum_{k = 0}^{\infty} \left( \int_{\frac{1}{2}}^{1} \frac{1}{z \ln 2} \left( \frac{1 - z}{z} \right)^{2k} \, \mathrm{d}z \right) x^{2k} (1 - x^2)\\
&=&\sum_{k = 0}^{\infty} I_k x^{2k} (1 - x^2),
\end{eqnarray}
where $I_k$ is as defined in~\eqref{eq:Ik}.

For every $k$, $I_k$ is analytically computable, e.g., $I_0 = 1$, $I_1 = 1 - \frac{1}{2 \ln 2}$, $I_2 = 1 - \frac{7}{12 \ln 2}$, etc. Since $I_k\geq0$ for all $k$, truncating the infinite sum gives the series of polynomial bounds~\eqref{eq:Phin}.
\end{proof}

This yields the following lower bound on the objective function.  
\begin{lemma}\label{lemm: objective_polynomial}
Let $a > 0$ and $\theta \in [0 , 2 \pi]$. Then
\begin{eqnarray}
 \Phi(a \cos\theta) - \Phi(a) \geq  \sum_{k=1}^{n} C_{k} a^{2k} \big(1 - \cos^{2k} \theta \big),
\end{eqnarray}
where $C_k:=I_{k-1}-I_k$.
\end{lemma}

\begin{proof}
From Lemma~\ref{lem: poly-approx}, we have the series expansion
\begin{eqnarray}
    \Phi(x) = I_0 + \sum_{k=1}^{\infty} \left(I_k - I_{k-1} \right) x^{2k} .
\end{eqnarray}

Hence, for any $a$ and $\theta$,
\begin{eqnarray}\label{eq:sum}
 \Phi(a\cos\theta)-\Phi(a)=\sum_{k=1}^{\infty} \left(I_{k-1}-I_k\right) a^{2k} \big(1 - \cos^{2k} \theta \big) .
\end{eqnarray}

Since $(1-z)/z\leq 1$ for $z\in[1/2, 1]$, $I_k$ is decreasing in $k$, and hence $C_k:=I_{k-1}-I_k\geq 0$. Each term in the sum in~\eqref{eq:sum} is therefore positive. Truncating the series at $k=n$ gives the claimed lower bound.
\end{proof}

After using this lemma the optimization problem~\eqref{eqn: optimization_final} is still not a polynomial problem, since $\cos(\xi_x)$ and $\sin(\xi_x)$ are not polynomials. To get around this we define
\begin{eqnarray}
\lambda_{1,x} := \tilde{a}_x \cos(\xi_x), \qquad 
\lambda_{2,x} := \tilde{a}_x \sin(\xi_x)
\end{eqnarray}
and treat $\lambda_{i,x}$ as free variables with the additional constraint
\begin{eqnarray}
\sum_{i} \lambda_{i,x}^2 = \tilde{a}_x^2.
\end{eqnarray}
Doing so allows us to lower bound \eqref{eqn: optimization_final} by a polynomial optimization problem. Using Lemma~\ref{lemm: objective_polynomial}, the objective function can be lower bounded using
\begin{eqnarray}
\Phi(\tilde{a}_x \cos \xi_x) - \Phi(\tilde{a}_x)
&=& \sum_{n=1}^{k} C_n \, \tilde{a}_x^{2n} \, \sin^2 \xi_x \, \left(1 + \cos^{2n-1} \xi_x \right) \\ 
&=& \sum_{n=1}^{k} C_n \, (\tilde{a}_x \sin \xi_x)^2 \left( \tilde{a}_x^{2n-2} + (\tilde{a}_x \cos \xi_x)^{2n-2} \right),
\end{eqnarray}
which can be written as some polynomial $P_n(\tilde{a}_x, \lambda_{1,x}, \lambda_{2,x})$.

Lower bounds on $\entropy_{p_{X}}(\overlap, \score)$ can now be obtained by solving this polynomial optimization problem using sum-of-squares (SOS) relaxation with software such as {\tt Ncpol2sdpa}. The final optimization problem takes the form
\begin{eqnarray}\label{eqn: final optimization problem for semi-DI}
\begin{aligned}  
\entropytwo(\score, \overlap) := 
     \inf \quad &  \sum_{x} p_X(x) \, P_n(\tilde{a}_x, \lambda_{1,x}, \lambda_{2,x})  \\
\textrm{s.t.} \quad &  \sum_{x} \left( -\eta_x + (-1)^x \lambda_{1,x} \right) \geq 4 \score - 4,  \\
              &  \left( \sum_{x} \lambda_{1,x} \right)^2 + \left( \sum_{x} \lambda_{2,x} \right)^2 \geq \left( 4 \overlap - \sum_{x} \eta_x \right)^2, \\ 
              &  \lambda_{1,x}^2 + \lambda_{2,x}^2 = \tilde{a}_x^2, \\ 
              &  0\leq\tilde{a}_x\leq\eta_x\leq1.
\end{aligned}       
\end{eqnarray}

Note that this is a constrained polynomial optimization problem over $8$ real variables. In fact, only $6$ are independent, with the remaining two being auxiliary (dummy) variables.

The final result of this appendix is summarized in the following theorem:
\begin{theorem}
Let $\entropytwo(\score, \overlap)$ be defined as in \eqref{eqn: final optimization problem for semi-DI}. Then
\begin{eqnarray}
F_{p_X}(\score, \overlap) \geq \mathrm{convenv}\left( \entropytwo(\score, \overlap) \right).
\end{eqnarray}
\end{theorem}

\begin{proof}
This is a consequence of the relaxations obtained via Lemma~\ref{lemm: objective_polynomial}, Lemma~\ref{eqn : Non-polynomial optimization problem}, and Lemma~\ref{lemm: F_to_G}.
\end{proof}

\section{Convex envelope}\label{app: LF transform}

In this section, we discuss the computation of the convex envelope of a function $f:\M{D}\to\mathbb{R}$, where $\M{D} \subseteq \mathbb{R}^n$ is a closed convex set. We start with the definition.
\begin{definition}[Convex envelope] 
 The convex envelope of $f$, denoted $\mathrm{convenv}(f):\M{D}\to\mathbb{R}$, is the function 
\begin{eqnarray}
    \mathrm{convenv}(f)(\B{x}) := \max_g \{ g(\B{x}) : g \ \mathrm{ is} \ \mathrm{ convex} \ \mathrm{ and }\ g(\B{x}')  \leq f(\B{x}')\ \forall\ \B{x}'\in\M{D}\} .
\end{eqnarray}
\end{definition}
In other words, $\mathrm{convenv}(f)$ is the convex function that is not greater than $f$ at any point, but otherwise has the largest possible value at every point. Alternatively, $\mathrm{convenv}(f)$ is the solution of the optimization problem 
\begin{eqnarray}\label{def:convenv}
\begin{aligned} 
  \mathrm{convenv}(f)(\B{x})=\inf_{\mu} \quad &  \int  f(\B{x}')\,\mathrm{d}\mu(\B{x}')\\ 
                    \textrm{s.t.} \quad &\int\B{x}'\,\mathrm{d} \mu(\B{x}')=\B{x},
 \end{aligned}
\end{eqnarray}
where the infimum is taken over all probability measures on $\M{D}$~\cite{Convex_envelope_book}. 

We now define an important tool in convex analysis that allows us to compute the convex envelope of a function.
\begin{definition}[Legendre-Fenchel transform] The Legendre-Fenchel transform of $f$ is $f^*: \mathbb{R}^{n} \to \mathbb{R}$ is defined by
\begin{eqnarray}
    f^*(\mathbf{k}) := \sup_{\mathbf{x} \in \M{D}} \left(\B{k}.\B{x} - f(\B{x}) \right) .
\end{eqnarray}
\end{definition}
The usefulness of this is given in the following lemma whose proof can be found in textbooks of convex analysis such as~\cite{Convex_envelope_book}.
\begin{lemma}
Let $f: \M{D} \to \mathbb{R}$ be bounded. Then the following holds 
\begin{itemize}
    \item $f^{*}$ is convex
    \item $(f^{*})^*(\B{x}) = \mathrm{convenv}(f)(\B{x})$ for all $\B{x}\in\M{D}$.
\end{itemize}
\end{lemma}
There are algorithms in available in the literature~\cite{lucet1997faster,contento2015linear} to compute the convex envelopes by computing the Legendre-Fenchel conjugate of the function twice. We compute the convex envelope using the method of~\cite{contento2015linear}, the code for which was generously provided by the authors to us.

\section{Monotonicity of the objective function}\label{app: useful claims}
\begin{lemma}\label{lem:16}
For all $\theta \in \mathbb{R}$, the function $f: [0 , 1] \times \mathbb{R} \to \mathbb{R}$ defined by $f(x, \theta)  = \Phi(x\cos\theta)-\Phi(x)$ is nondecreasing in $x$, and strictly increasing unless $x = 0$ or $\sin(\theta) = 0$.
\end{lemma}
\begin{proof}
We start with the derivative:
\[
\frac{\partial}{\partial x}\big(\Phi(x\cos\theta)-\Phi(x)\big)
= \cos\theta\,\Phi'(x\cos\theta)-\Phi'(x).
\]
Using the identity $\Phi'(z)=-\frac{1}{\ln(2)}\int_{0}^{1}\frac{z\,du}{1-z^{2}u^{2}}$ for $|z|<1$, we obtain
\[
\frac{\partial}{\partial x}\big(\Phi(x\cos\theta)-\Phi(x)\big)
= \frac{\sin^{2}\theta}{\ln(2)} \int_{0}^{1} \frac{x }{(1-x^{2}u^{2})(1-x^{2}\cos^{2}\theta\,u^{2})}\,\mathrm{d}u\ge0,
\]
with equality if and only if $\sin(\theta)=0$ or $x=0$.
\end{proof}

The next lemma concerns the function $g:[0,1]\times[0,1]\times[0,2\pi]\to\mathbb{R}$ defined by $g(x,\eta,\theta)=\eta f(x,\theta)-f(\eta x,\theta)$, where $f(x,\theta)$ is the function defined in Lemma~\ref{lem:16}.
\begin{lemma}\label{Claim: monotonicity in x}
The function $g(x,\eta,\theta)$ is non-negative.
\end{lemma}
\begin{proof}
    From the definition we have $g(0,\eta,\theta)=0$ for all $\eta$ and $\theta$. It hence suffices to show $\frac{\partial g}{\partial x}\geq0$ for all $x$, $\eta$ and $\theta$. We can compute
    \begin{align*}
        \frac{\partial g}{\partial x}&=\eta\frac{\partial g}{\partial x}(x,\theta)-\eta\frac{\partial g}{\partial x}(\eta x,\theta)\\
        &=\eta\frac{\sin^2(\theta)}{ \ln(2)}\int_0^1\frac{x}{(1-x^2u^2)(1-x^{2}\cos^2(\theta)u^2)}-\frac{\eta x}{(1-(\eta x)^2u^2)(1-(\eta x)^2\cos^2(\theta)u^2)}\,\mathrm{d}u\\
        &\geq0,
    \end{align*}
where the last line follows because $x\geq\eta x$, $(1-x^2u^2)^{-1}\geq(1-(\eta x)^2u^2)^{-1}$ and $(1-x^2\cos^2(\theta)u^2)^{-1}\geq(1-(\eta x)^2\cos^2(\theta)u^2)^{-1}$, so the integrand is positive.
\end{proof}

\section{Entropy accumulation theorem}\label{app: EAT} 
To derive the rates, we use the version of the Entropy Accumulation Theorem (EAT) presented in~\cite{LLR&,BhRC}. Specifically, we apply the formulation given in Theorem~3 of~\cite{BhRC}. In this framework, the EAT is expressed in terms of a collection of EAT channels $\{\mathcal{N}_i\}_i$, where each channel is of the form
\[
\mathcal{N}_i : \hat{R}_{i-1} \to \hat{R}_iA_iD_iU_i.
\]
These registers often have the following associations (although we will use them in different ways below):
\begin{itemize}
    \item $A_i$ is a classical register representing the outputs of the protocol in round $i$
    \item $D_i$ is a classical register representing the inputs of the protocol in round $i$
    \item $\hat{R}_{i-1}$ and $\hat{R}_i$ denote the systems held by the devices before and after round $i$
    \item $U_i$ is a classical register that stores a deterministic function of $A_i$ and $D_i$, usually a score.
\end{itemize}

To apply the EAT, it must hold that
\[
I(A_1^{i-1} : D_i \mid D_1^{i-1}, E) = 0,
\]
where \(A_1^{i-1} = A_1 A_2 \ldots A_{i-1}\), \(D_1^{i-1} = D_1 D_2 \ldots D_{i-1}\), and \(I\) denotes the conditional mutual information.
This condition requires that the inputs of the current round are independent of the outputs of all previous rounds, given the previous inputs, the previous outputs, and any side information \(E\).

The EAT then provides a bound on the conditional smooth min-entropy of the outputs $\mathbf{A} = (A_1, A_2, \dots, A_n)$ conditioned on the inputs $\mathbf{D} = (D_1, D_2, \dots, D_n)$ and side information $E$, i.e., it bounds $H^{\epsilon}_{\min}(\mathbf{A} | \mathbf{D} E)_{\rho}$ for $\epsilon>0$, 
where $\rho = \mathcal{N}_n \circ \mathcal{N}_{n-1} \circ \cdots \circ \mathcal{N}_1(\rho_{\hat{R}_0 E})$, for a collection of EAT channels $\{\mathcal{N}_i\}_{i=1}^n$ and an initial state $\rho_{\hat{R}_0 E}$.

Roughly speaking, the EAT gives a bound of the form 
\[
H^{\epsilon}_{\min}(\mathbf{A}|\mathbf{D} E)_{\rho} \geq n r - \sqrt{n} v,
\]
where $r$ is a lower bound on the conditional von Neumann entropy of a representative round under an EAT channel, and $v$ is an error term.

For our protocol, the single-round channels are defined as
\[
\mathcal{N}_i : R_i C_i \to T_iX_iY_iR_{i+1}C_{i+1}.
\]
We define an additional variable $U_i = (T_i , X_i , Y_i)$. We then map these to the registers of the EAT as follows:
\begin{itemize}
    \item \textbf{Protocol~\ref{protocol: Semi-DI recycling}} (recycled inputs): $A_i = (T_i, X_i, Y_i)$, and $D_i$ trivial.
    \item \textbf{Protocol~\ref{protocol: private to public}} (inputs not recycled): $A_i = Y_i$, and $D_i = (X_i, T_i)$.
\end{itemize}

In both cases, $U_i$ is a deterministic function of $A_i$ and $D_i$. In the first case, the condition $I(A^{i-1} : D_i|D^{i-1}, E) = 0$ trivially holds; in the second case, it holds because of the assumption that $T_i$ and $X_i$ are chosen randomly in each round. Identifying $\hat{R}_i \equiv (R_i, C_i)$, we conclude that $\mathcal{N}_i$ is an EAT channel.

The EAT then provides a bound on the conditional min-entropy $H_{\min}(\mathbf{A}|\mathbf{D} E)_{\rho}$, where $\rho = \mathcal{N}_n \circ \mathcal{N}_{n-1} \circ \cdots \circ \mathcal{N}_1(\rho_{\hat{R}_0 E})$, for a collection of EAT channels $\{\mathcal{N}_i\}_{i=1}^n$ and an initial state $\rho_{\hat{R}_0 E}$.

\subsection{Min-tradeoff function} 
A min-tradeoff function is an affine function that is always below the rate function, where by rate function we mean a lower bound on the single-round conditional von Neumann entropy over all EAT channels and input states that give the observed distribution $q$ over $U=(T,X,Y)$. Given such a distribution the score overlap and test probability are implied:
\[
\score_q = \frac{q(0,0,0)}{2(1-\gamma)p_X(0)} + \frac{q(0,1,1)}{2(1-\gamma)p_X(1)}, \quad
\overlap_q = \frac{q(1,0,0)}{2\gamma p_X(0)} + \frac{q(1,1,0)}{2\gamma p_X(1)}, \quad
\gamma_q = \sum_{x,y} q(1,x,y),
\]
and the rate function can be taken as 
\[
\rate(q) = (1-\gamma_q)F_{p_X}(\score_q, \overlap_q).
\]
Note that $\gamma_q$ is the probability of a generation rounds and so equals $\gamma$ by the design of the protocol. Any affine function $\tilde{f}(\gamma,\score,\overlap)$ for which $f(q):=\tilde{f}(\gamma_q,\score_q,\overlap_q)$ lower bounds $\rate(q)$ is a valid min-tradeoff function. Given a min-tradeoff function $f(q)$ the quantities $\mathrm{Max}[f]$, $\mathrm{Min}_{\M{Q}}[f]$ and $\mathrm{Var}[f]$ are of interest. The latter is the variance of $f$, while the first two are defined by
\begin{align*}
\mathrm{Max}[f]&= \max_{\gamma,\score,\overlap \in [0,1]} \tilde{f}(\gamma,\score,\overlap),\\
\mathrm{Min}_{\mathcal{Q}}[f]&=\min_{(\gamma,\score,\overlap)\in\mathcal{Q}} \tilde{f}(\gamma,\score,\overlap),
\end{align*}
where $\mathcal{Q}$ is the set of $(\gamma,\score,\overlap)$ allowed in quantum theory.

Finally, we can bound the variance using the Bhatia–Davis inequality~\cite{BhatiaDavis}:
\[
\mathrm{Var}[f] \le (M - \mu)(\mu - m),
\]
where $M$ is the maximum value of $f$, $m$ is the minimum, and $\mu$ is the expectation, as done in~\cite{BhRC}.   

We consider computing $F_{p_{X}}(\score , \overlap)$ on a discrete grid induced by a gridding $\M{P}$ of the domain $[0, 1 ] \times [0 ,1] $ \footnote{we set $F_{p_{X} }(\score , \overlap) = 0$, whenever $\score \leq 1/2$ or $\overlap \leq 1/2$.} of $(\score,\overlap)$ used to compute the rate function. We can then construct a gridding of the set $\tilde{\M{P}}$ of the domain $[0 ,1]^{\times 3}$ of $\tilde{f}$ to obtain the values of the function $(1-\gamma) F_{p_{X}}(\score , \overlap)$ on a grid of points. 

\subsubsection{General problem}

We discuss here how to construct a min-tradeoff function based on a grid. Let
\(
\mathcal{C} := [a^{(1)},b^{(1)}] \times [a^{(2)}, b^{(2)}] \times \dots \times [a^{(n)}, b^{(n)}] \)
be a (hyper) rectangle, and let $\mathcal{P}$ be a gridding of $\mathcal{C}$ into a finite grid with the grid sets $\M{X}^{(i)}=\{x_i^{(i)},\ldots,x_{n_i}^{(i)}\}$ as in Appendix~\ref{app:grid}. Suppose \(H : \mathcal{C} \to \mathbb{R}\) is coordinate-wise non-decreasing, i.e.,
\[
H(\mathbf{y}) \ge H(\mathbf{x}) \quad \text{whenever } \mathbf{y} \ge \mathbf{x},
\]
where $\mathbf{y} = (y_1, \dots, y_n) \ge \mathbf{x} = (x_1, \dots, x_n)$ means $y_i \ge x_i$ for all $i$. 

\subsubsection{Convex Lower Bound.}
Let $g : \mathcal{C} \to \mathbb{R}$ be a convex function such that for every grid point $\mathbf{x} \in \mathcal{P}$,
\[
g(\mathbf{x})\begin{cases}\leq H(\mathbf{x}_{-})&\text{if }\mathbf{x}_{-}\text{ exists}\\=c&\text{otherwise}\end{cases}\,,
\]
where for $\mathbf{x}=(x_{i_1}^{(1)},x_{i_2}^{(2)},\ldots,x_{i_n}^{(n)})$, we define $\mathbf{x}_{-} := (x^{(1)}_{i_1-1}, x^{(2)}_{i_2-1},\ldots,x^{(n)}_{i_n-1})$ whenever the predecessor indices exist (i.e., if $i_j\geq2$ for all $j$), and $c < \min_{\mathbf{x} \in \mathcal{C}} H(\mathbf{x})$ is a constant.

Let ${\mathbf{i}}=(i_1,i_2,\ldots,i_n)$ For any $\mathbf{x}$ in the hyper-rectangle
\[
R_{\mathbf{i}}:=[x_{i_1}^{(1)}, x_{i_1+1}^{(1)}] \times [x_{i_2}^{(2)}, x_{i_2+1}^{(2)}] \times \ldots \times [x_{i_n}^{(n)}, x_{i_n+1}^{(n)}],
\]
convexity implies
\begin{eqnarray}
g(\mathbf{x}) &\le& \max_{\mathbf{y}\in\mathrm{Vertices}(R_{\mathbf{i}})}g(\mathbf{y})\\
&\le& \max_{\mathbf{y}\in\mathrm{Vertices}(R_{\mathbf{i}})} H(\mathbf{y}_-)\\
&=& H(x_{i_1}^{(1)}, x_{i_2}^{(2)}, \dots, x_{i_n}^{(n)}),
\end{eqnarray}
where the last equality follows since $H$ is non-decreasing in every coordinate.

Thus, for any $\mathbf{x} \in \mathcal{C}$,
\[
H^{\mathcal{P}}(\mathbf{x}) - g(\mathbf{x}) = H(\lfloor \mathbf{x} \rfloor_{\mathcal{P}}) - g(\mathbf{x}) \ge 0,
\]
so $g$ is a valid lower bound on $H^{\mathcal{P}}$.

\subsubsection{Affine Min-Tradeoff Function.}
We define the min-tradeoff function to be any affine function $f$ satisfying
\[
H(\mathbf{x}) \ge f(\mathbf{x}_{-})\quad\forall \ \mathbf{x}\in\mathcal{P}.
\]
A good choice of $f(\mathbf{x}) = \mathbf{c}.\mathbf{x} + d$ can be computed by solving the linear program. Let $\B{x}^*$ be the observed experimental statistics, we solve the following optimization problem:
\begin{align*}
\min_{\mathbf{c},d}\ & H(\mathbf{x}^*) - f(\mathbf{x}^*) \\
\text{s.t.}\ & H(\mathbf{x}) \ge f(\mathbf{x}_{-}) \quad \forall\ \mathbf{x} \in \mathcal{P}.
\end{align*}

\subsection{Completeness error}
We now compute a bound on the probability that an honest protocol aborts. We use the same method as in~\cite[Section E.3.2]{BhRC}. This relies on Hoeffding's inequality~\cite{Hoeffding}, i.e., that if $\{X_i\}_{i=1}^n$ are i.i.d.\ random variables with $a\leq X_i\leq b$, $S=\sum_i X_i$ and $\mu=\mathbb{E}(S)$, then for $t>0$,
\[\mathbb{P}(S-\mu\geq t)\leq e^{-\frac{2t^2}{n(b-a)^2}}.\]

For the completeness error we want the probability of $\overlap_\#<\overlap_{\mathrm{exp}}-\delta_{\overlap}$ and $\score_\#<\score_{\mathrm{exp}}-\score_{\overlap}$. We consider both of these separately.

Recall that
\begin{eqnarray}
\overlap_\# := \frac{1}{2}\sum_{x} \frac{|\{i : U_{i} = (1 , x , 0) \}|}{n p_{X}(x)\gamma}, \quad\score_\# := \frac{1}{2}\sum_{x}\frac{|\{i :U_{i} = (0 , x , x)\}|}{n p_{X}(x)(1 - \gamma) }  \nonumber .
\end{eqnarray}
To put these in a suitable form to apply Hoeffding's inequality, consider
\begin{equation}
    Z_{i} = \begin{cases} \frac{1}{2n\gamma p_{X}(x)}  &\text{if } T_i=1,X_i=x,Y_i=0\\ 
    0       & \text{otherwise}
    \end{cases}
\end{equation}
and
\begin{eqnarray}
    W_{i} = \begin{cases} \frac{1}{2n(1 - \gamma) p_{X}(x)}  &\text{if } T_i=0,X_i=x,Y_i=x\\ 
    0       & \text{otherwise}
    \end{cases}
\end{eqnarray}
so that $\overlap_\#=\sum_iZ_i$ and $\score_\#=\sum_iW_i$. In an honest implementation of the protocol, we have
\begin{eqnarray}
 \mathbb{P}(U_i = (t_i , x_i , y_i ))  = p_{T}(t_i) p_{X}(x_i) p_{Y|t_i,x_i}(y_i). 
\end{eqnarray} 
and we set $\overlap_{\mathrm{exp}}=\mathbb{E}(\overlap_\#)$ and $\score_{\mathrm{exp}}=\mathbb{E}(\score_\#)$. We have
\begin{eqnarray}
    \mathbb{P}\left[\sum_{i} Z_{i} \leq\overlap_{\mathrm{exp}} - \delta_{\overlap} \right] &=&\mathbb{P}\left[\sum_{i}(-Z_{i})-(-\overlap_{\mathrm{exp}}) \geq\delta_{\overlap}\right],
\end{eqnarray}
where in effect we are taking random variables $-Z_i$ with expectation $-\overlap_{\mathrm{exp}}$. Since $-\frac{1}{2n\gamma\min_x[P_X(x)]}\leq(-Z_i)\leq0$, using Hoeffding's inequality gives
\begin{eqnarray}
    \mathbb{P}\left[\sum_{i} Z_{i}\leq\overlap_{\mathrm{exp}} - \delta_{\overlap} \right] &\leq&e^{-8n (\delta_{\overlap}\gamma p_{X}(0))^2 },
\end{eqnarray}
where we take $P_X(0)\leq P_X(1)$ (without loss of generality).
Similarly,
\begin{eqnarray}
    \mathbb{P}\left[\sum_{i} W_{i} \leq\score - \delta_{\score}\right] = e^{-8n(\delta_{\score}(1-\gamma)p_{X}(0))^2 }
\end{eqnarray}

We need to compute the probability of the event that both $\sum_iZ_{i} \leq \overlap - \delta_{\overlap}$ and $\sum_iW_{i} \leq \score - \delta_{\score}$. Using the union bound we have that the probability that the protocol does not abort satisfies
\begin{eqnarray}
     p_{\bar{\Omega}} \leq e^{-8n(\delta_{\overlap} \gamma p_{X}(0) )^2 } +e^{-8n(\delta_{\score} (1 - \gamma) p_{X}(0) )^2 }. 
\end{eqnarray}
Given a desired completeness error $\epsilon_C$, we choose $\delta_{\score} > 0$ and $\delta_{\overlap} > 0$ so that $F_{p_{X}}(\score - \delta_{\score} , \overlap - \delta_{\overlap})$ is maximized.

\section{Remarks on Assumption \ref{ass: carrier}}\label{app: discussion_on_carrier_information} 

This appendix addresses Assumption~\ref{ass: carrier} and implementation-level measures that can reduce the risk that information about $X$ is conveyed via unintended side channels, acknowledging that such measures do not completely eliminate the problem. In particular, if the source is optical and emits additional signals beyond the intended optical band, then these signals may still be photonic (e.g., out-of-band light, parasitic spectral components, or other unmonitored optical degrees of freedom) and hence propagate at (or extremely close to) the speed of light. In the presence of such additional signals, timing- or velocity-based arguments are not sufficient to exclude side-channel leakage, especially once realistic processing and detection latencies are included.

Accordingly, we treat Assumption~\ref{ass: carrier} as an implementation constraint that can be enforced approximately, using countermeasures that reduce the accessible side-channel manifold. Concretely, the admissible optical band can be restricted using well-characterised spectral filtering, dichroic elements, and wavelength-selective coatings, together with spatial-mode constraints such as single-mode coupling and apertures, thereby strongly attenuating emission outside the intended band and mode. In addition, it is natural to complement the primary power meter with broader monitoring (e.g., a broadband photodetector and/or a spectrally resolving monitor on a tap of the outgoing light) and to impose an explicit abort condition whenever out-of-band energy or related statistics exceed predefined thresholds. Since side-channel leakage can manifest through changes in temporal structure even at fixed total power, one can monitor the distribution of detection times and inter-pulse timing at trusted monitoring points and abort upon statistically significant deviations from calibrated behaviour, which also constrains time-encoding strategies. We restrict the discussion to these high-level mitigations because a comprehensive treatment of hidden photonic side channels is inherently system-specific and lies outside the scope of the present security proof.

\bibliography{ref}

\end{document}